\PassOptionsToPackage{safe}{pbalance}
\documentclass[sigconf,nonacm,pbalance]{acmart}

\newcommand{\sysattack}{\textsc{Collapse}\xspace}
\newcommand{\sysdefense}{$\mathcal{O}_\text{ext}$\xspace}
\newcommand{\priorboundary}{$\mathcal{O}_\text{prior}$\xspace}
\newcommand{\extboundary}{$\mathcal{O}_\text{ext}$\xspace}

\providecommand{\nnz}{\operatorname{nnz}}
\providecommand{\suppop}{\operatorname{supp}}
\providecommand{\colop}{\operatorname{col}}

\newcommand{\eg}{\hbox{{e.g.,}}\xspace}
\newcommand{\ie}{\hbox{{i.e.,}}\xspace}

\newcommand{\parab}[1]{\noindent{\bf #1}\xspace}

\usepackage{graphicx}
\usepackage{placeins}

\usepackage{subcaption} 

\usepackage{tabularx} 

\usepackage{booktabs} 

\usepackage{multirow} 

\usepackage{xspace}

\usepackage{enumitem}

\usepackage{makecell} 

\usepackage{amsmath}
\usepackage{amsthm}
\usepackage{pifont}

\theoremstyle{plain}
\newtheorem{theorem}{Theorem}

\theoremstyle{definition}

\newtheorem{definition}{Definition}

\theoremstyle{remark}

\usepackage{tcolorbox}

\newtcolorbox{challengebox}{
    colback=gray!10,      
    colframe=black,       
    boxrule=0.5pt,        
    arc=0pt,              
    boxsep=2pt,           
    left=4pt, right=4pt, top=4pt, bottom=4pt, 
    sharp corners,        
    halign=left,          
    fontupper=\small      
}
\newtcolorbox{insightbox}[1]{
    colback=gray!6,
    colframe=black,
    boxrule=0.5pt,
    arc=1pt,
    title={#1},
    fonttitle=\bfseries
}

\usepackage[ruled,vlined,linesnumbered]{algorithm2e}

\AtBeginDocument{%
  }

\copyrightyear{2026}
\acmYear{2026}
\setcopyright{cc}
\setcctype{by}
\acmConference[CCS '26]{Proceedings of the 2026 ACM SIGSAC Conference on Computer and Communications Security}{November 15--19, 2026}{The Hague, Netherlands}
\acmBooktitle{Proceedings of the 2026 ACM SIGSAC Conference on Computer and Communications Security (CCS '26), November 15--19, 2026, The Hague, Netherlands}
\acmDOI{10.1145/3830454.3846712}
\acmISBN{979-8-4007-2871-6/2026/11}

\begin{document}

\title{Understanding the Security Boundary of Obfuscation-based On-Device LLM Protection} 

\author{Hanyi Zhou}
\affiliation{%
  \institution{Tsinghua University}
  \city{Beijing}
  \country{China}}
\email{hy-zhou24@mails.tsinghua.edu.cn}

\author{Chenyang Li}
\affiliation{%
  \institution{Tsinghua University}
  \city{Beijing}
  \country{China}}
\email{lcy23@mails.tsinghua.edu.cn}

\author{Yuanzhe Pang}
\affiliation{%
  \institution{Tsinghua University}
  \city{Beijing}
  \country{China}}
\email{pangyz24@mails.tsinghua.edu.cn}

\author{Ke Xu}
\affiliation{%
  \institution{Tsinghua University}
  \city{Beijing}
  \country{China}}
\email{xuke@tsinghua.edu.cn}

\author{Mingwei Xu}
\affiliation{%
  \institution{Tsinghua University}
  \city{Beijing}
  \country{China}}
\email{xumw@tsinghua.edu.cn}

\author{Zhuotao Liu}
\authornote{Zhuotao Liu is the corresponding author.}
\affiliation{%
  \institution{Tsinghua University}
  \city{Beijing}
  \country{China}}
\email{zhuotaoliu@tsinghua.edu.cn}

\renewcommand{\shortauthors}{Hanyi Zhou et al.}

\begin{abstract}
Trusted Execution Environments (TEEs) offer a promising mechanism for safeguarding the intellectual property of on-device Large Language Models (LLMs). 
To overcome the inherent computational bottlenecks of TEEs, existing TEE-Shielded LLM Partition (TSLP) methods apply efficient obfuscation schemes to computationally intensive layers, offloading them to external GPUs while retaining only lightweight operations within the TEE. 
Although a growing body of TSLP-based approaches has emerged, these defense mechanisms remain largely heuristic. Consequently, some methods are proven vulnerable to certain specialized adversarial attacks designed to exploit their specific architectural implementations.

To overcome the limitations of these heuristic designs, this paper addresses a fundamental research question: can we establish common primitives to unify representative prior methodologies, characterize the security boundary of their compositions, and systematically extend them?
To this end, we formalize a set of obfuscation primitives, defined as dual-tuples of linear computations satisfying specific algebraic properties. 
We demonstrate that the matrix-level weight transformations of several representative efficient TSLP frameworks can be expressed as compositions of these primitives; consequently, the canonical form of these primitive compositions, denoted as \priorboundary, defines the security boundary of this primitive family. 
We then expose the vulnerabilities of \priorboundary through a novel primitive-guided attack methodology, \sysattack, demonstrating a shared vulnerability in several prominent TSLP methods published in top-tier venues, such as ArrowCloak (Security'25), TSQP (S\&P'25), and LoRO (NeurIPS'25).
Finally, we introduce two novel obfuscation primitives and integrate them with existing constructs to formulate \sysdefense, extending the prior security boundary \priorboundary.
Extensive experiments demonstrate that \sysattack can extract high-fidelity surrogate models from prior defenses, whereas \sysdefense reduces surrogate accuracy to the empirical black-box reference on average.

\end{abstract}




\maketitle

\section{Introduction}
Large language models (LLMs)~\cite{vaswani2017attention,brown2020gpt3,touvron2023llama} are increasingly deployed on end devices
to reduce inference latency while keeping user data on-device~\cite{gunter2024appleintelligence,qwen2024qwen25,liu2024mobilellm,geminiteam2023gemini}.
Because these models are valuable proprietary assets, 
placing the model weights on end devices raises significant concerns regarding model theft.
Compared to cloud-based LLM serving, a local adversary could obtain white-box visibility into code, parameters, memory, and intermediate activations of the on-device LLMs that the cloud-based APIs would never expose, resulting in model-extraction~\cite{tramer2016stealing,orekondy2019knockoff,krishna2020thieves,jagielski2020high,carlini2024stealing,nayan2024sok} and privacy attacks~\cite{shokri2017membership,fredrikson2015model,wang2018stealing}.

Trusted Execution Environments (TEEs) such as Intel~SGX~\cite{costan2016sgx} and Arm~TrustZone~\cite{pinto2019trustzone} provide hardware-isolated execution even when the operating system is untrusted, making them suitable for securing the on-device models. 
However, commodity TEEs are not scalable to host full LLM inference on their own: secure memory is limited, and their compute throughput falls far short of what modern LLM serving requires.
Recent works therefore adopt a TEE-Shielded LLM Partition (TSLP) methods~\cite{tramer2019slalom,mo2020darknetz,lee2019occlumency,shen2022soter,li2024translinkguard,wang2025arrows,xiong2025loro,zhou2023nnsplitter}, in which the TEE computes the lightweight non-linear computations,
(e.g., $\operatorname{LayerNorm}$, $\operatorname{Softmax}$), while the compute-intensive linear operators (\ie matrix multiplications) are offloaded to the surrounding Rich Execution Environments (REEs), such as the commodity GPU/NPU runtime. 
As these REEs are not trusted to be secure, the TEEs are required to obfuscate the model weights (for instance via various matrix manipulation tricks) before offloading them to REEs.
The REEs then compute on the obfuscated inputs, and the TEE subsequently recovers the correct results by ``subtracting'' the obfuscation factors from the results. 

Existing research on TSLP has developed along two parallel lines: developing new obfuscation schemes and exploiting the vulnerabilities of these schemes. 
However, existing obfuscation methodologies are largely ad-hoc.
Each obfuscation scheme invents tricks based on
certain efficiently-computable algebraic operations, such as column-wise permutation~\cite{li2024translinkguard}, column-wise scaling~\cite{shen2022soter,sun2025tsqp}, additive masking~\cite{tramer2019slalom,zhou2023nnsplitter,sun2023shadownet}, or low-rank masking~\cite{xiong2025loro}, and applies them to masking the offloaded linear layers.
In each case, the authors show that a handful of obvious inversions fail against the chosen tricks and declare their proposed scheme secure. 
Meanwhile, attacks against these mechanisms are likewise ad-hoc, where each targets a specific class of obfuscation methods. Zhang~et~al.~\cite{zhang2024noprivacy} target schemes built on additive masking~\cite{tramer2019slalom,zhou2023nnsplitter}, and ArrowMatch~\cite{wang2025arrows} targets schemes built on column-wise permutation and scaling~\cite{shen2022soter,li2024translinkguard}.

To overcome the limitations of these heuristic designs, this paper addresses a fundamental research question: can we establish common primitives to unify representative prior methodologies, characterize the security boundary of their compositions, and systematically extend them?
To this end, we formalize a set of
\emph{obfuscation primitives}, defined as dual-tuples of linear computations satisfying specific algebraic properties. 
We demonstrate that the matrix-level weight transformations studied in several prominent TSLP methods published in top-tier venues~\cite{zhou2023nnsplitter,sun2023shadownet,shen2022soter,li2024translinkguard,xiong2025loro,sun2025tsqp,zhang2024groupcover,wang2025arrows} can be expressed as compositions of these primitives.
Specifically, we identify two categories of primitives, multiplicative primitives (permutation $\Pi$ and scaling $D$) and additive primitives (sparse mask $S$ and low-rank mask $L$).
Each primitive is invertible and efficiently-computable, so any combination of them is also efficiently-computable within TEE's secure enclave. 
Given these unified primitives, we identify \priorboundary as the canonical form of primitive compositions of these primitives, which can represent the security boundary of existing TSLP methods.


Yet, we demonstrate that $\mathcal{O}_\textsf{prior}$ is vulnerable, revealing a shared vulnerability in these approaches.
Specifically, we develop a novel attack methodology \sysattack against $\mathcal{O}_\textsf{prior}$. 
Our key insight is that the column directions of the protected model weight $W_\textsf{vic}$ cannot be effectively masked by arbitrary composition of these primitives. 
Specifically, column-wise permutation $\Pi$ and column-wise scaling $D$ only permute and rescale columns, so the column directions of $\mathcal{M}_\textsf{vic}$ are preserved and remain correlated with the public $\mathcal{M}_\textsf{pre}$. 
Second, sparse mask $S$ only masks a tiny fraction of elements in $\mathcal{M}_\textsf{vic}$, so the column direction of $\mathcal{M}_\textsf{vic}$ are largely intact. 
Third, low-rank mask $L$ hides $\mathcal{M}_\textsf{vic}$ inside one low-dimensional subspace $\mathcal{V}_L$, so the column directions of $\mathcal{M}_\textsf{vic}$ are intact in the subspace orthogonal to $\mathcal{V}_L$. 
As a result, we can apply a systematic methodology to leverage such leaked information to obtain a surrogate model $\mathcal{M}_\textsf{sur}$ with the same architecture and comparable performance with $\mathcal{M}_\textsf{vic}$. 
At a high level, our attack methodology has three steps: 
aligning the permutations of multiple observations of $\mathcal{M}_\textsf{vic}$; combining the aligned observations to recover the low-rank mask $L$ and locating the positions that has been masked by the sparse mask $S$ to obtain an intermediate model $\mathcal{M}_\textsf{init}$; fine-tuning $\mathcal{M}_\textsf{init}$ on a small labeled dataset (collected by querying $\mathcal{M}_\textsf{vic}$) to remove the remaining sparse mask $S$ and scaling $D$, resulting in $\mathcal{M}_\textsf{sur}$.

To further extend this security boundary, we propose two novel obfuscation primitives: the \emph{sparse multiplicative primitive} mixes the columns of $\mathcal{M}_\textsf{vic}$ and \emph{double-sided multiplicative primitive} disturbs both column and row of $\mathcal{M}_\textsf{vic}$.
Combining these two primitives with existing primitives yields \sysdefense, which disrupts the row- and column-wise invariants on which \sysattack relies.
Extensive experiments across four representative models (\eg BERT-Base, ViT-Base, Qwen2.5-0.5B, and Qwen2.5-1.5B) demonstrate that while \sysattack can extract high-fidelity surrogate models secured by prior defenses, \sysdefense reduces the resulting surrogate accuracy to the empirical black-box reference on average.

We clarify that we do not claim the universal security of \sysdefense. In particular, we evaluate \sysdefense only against the existing
\sysattack pipeline, without a defense-specific adaptation. Thus, we do not claim the universal security of \sysdefense. Yet, our main contribution is to demonstrate a unified methodology to understand the security of obfuscation-based on-device LLM protection. Concretely, the contributions of our work are as follows.

\begin{itemize}[leftmargin=2em]
   \item \textbf{Unified Obfuscation Primitives for TSLP Approaches.} We formalize a set of obfuscation primitives that underpin the matrix-level weight transformations of several representative TSLP methods published in top-tier venues, providing a common algebraic abstraction for understanding and analyzing these approaches within a single framework.

   \item \textbf{Identify the Security Boundary.} Building on the unified analysis abstraction, we define $\mathcal{O}_{\text{prior}}$ as a canonical structural form of defense by composing existing primitives, which represents the security boundary of these methods.

   \item \textbf{Breaking the Security Boundary.} We develop \sysattack, a systematic, primitive-guided attack methodology against $\mathcal{O}_{\text{prior}}$ and its instantiations, demonstrating a common vulnerability shared by existing TSLP approaches.

   \item \textbf{Extending the Boundary.} Guided by \sysattack, we propose \sysdefense to further extend the security boundary beyond $\mathcal{O}_{\text{prior}}$ by introducing novel obfuscation primitives. We evaluate both \sysattack and \sysdefense across representative models, datasets, and obfuscation schemes.

 \end{itemize}

\section{Background and Threat Model}
\label{sec:background}

\subsection{Transformer Architecture and TEE--REE Collaborative Inference}
\label{sec:bg_tee_ree}

\subsubsection{Operator Classes in Transformer}
As illustrated in Fig.~\ref{fig:bg_transformer_partition}, a transformer block contains three operator classes that differ on what matters for on-device LLM protection---whether they carry proprietary weights and whether they dominate arithmetic cost. $\operatorname{Op}_{\text{linear}}$ comprises the weight-bearing linear transforms (QKV, output projection, feed-forward-network (FFN) up/down), which carry the proprietary parameters. $\operatorname{Op}_{\text{quadra}}$ comprises the token-to-token products inside self-attention, whose cost grows quadratically with sequence length. $\operatorname{Op}_{\text{non-poly}}$ collects the remaining lightweight non-linear operators (normalization, activation, and residual).

\begin{figure}[!htb]
  \centering
  \includegraphics[width=0.83\linewidth]{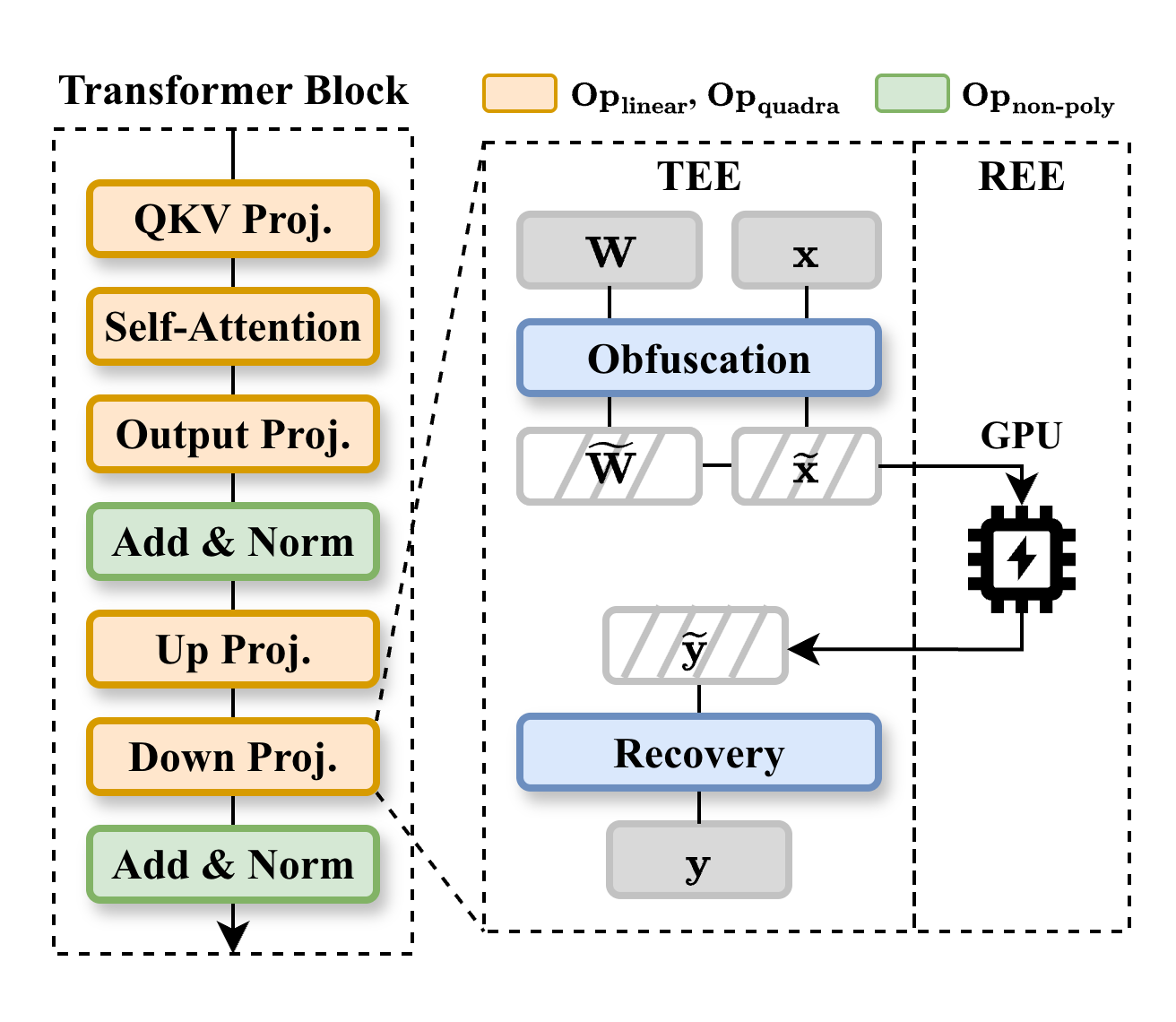}
  \Description{A transformer block whose QKV projection, self-attention, output projection, up projection, down projection, and Add and Norm stages illustrate three classes of inference operators. Beside it, a representative offloaded stage is shown flowing through obfuscation in the TEE, computation in the REE, and recovery back in the TEE.}
  \caption{Operator classes in a transformer block (left) and the obfuscate--compute--recover pipeline for computing $\operatorname{Op}_{\text{linear}}$ in TSLP methods (right).}
  \label{fig:bg_transformer_partition}
\end{figure}

\subsubsection{TEE--REE Collaborative Inference}
Trusted Execution Environments such as Intel~SGX~\cite{costan2016sgx} and Arm~TrustZone~\cite{pinto2019trustzone} enforce code and data isolation even when the OS is untrusted, but have limited secure memory and weaker throughput than the accelerators used for LLM inference. Recent systems therefore retain the lightweight $\operatorname{Op}_{\text{non-poly}}$ stages inside the TEE and offload $\operatorname{Op}_{\text{linear}}$ and $\operatorname{Op}_{\text{quadra}}$ to the untrusted Rich Execution Environment (REE)~\cite{tramer2019slalom,lee2019occlumency,mo2020darknetz,shen2022soter,li2024translinkguard,wang2025arrows,zhang2024noprivacy}. Before offloading, the weight-bearing linear layers are obfuscated. 
For input $\mathbf{X}$ and protected weight $\mathbf{W}$, the TEE releases the obfuscated $\widetilde{\mathbf{W}}=\mathcal{O}(\mathbf{W})$, the REE returns $\mathbf{X}\widetilde{\mathbf{W}}$, and the TEE applies a recovery map $\mathcal{R}$ satisfying $\mathcal{R}(\mathbf{X}\widetilde{\mathbf{W}})=\mathbf{X}\mathbf{W}$.

\subsection{Threat Model}
\label{sec:threat_model}
Figure~\ref{fig:bg_attack_surface} summarizes the threat model. We adopt an \emph{honest-but-curious} adversary model: the on-device LLM inference follows the TEE--REE protocol correctly (so that the recovered output $\mathcal{R}(\mathbf{X}\widetilde{\mathbf{W}})$ equals $\mathbf{X}\mathbf{W}$), but the device owner can record every distinct exposed view $\widetilde{\mathbf{W}}_t$ across key refreshes and attempts to reconstruct $\mathbf{W}$ from them together with public prior knowledge.

\begin{figure}[!htb]
  \centering
  \includegraphics[width=\linewidth]{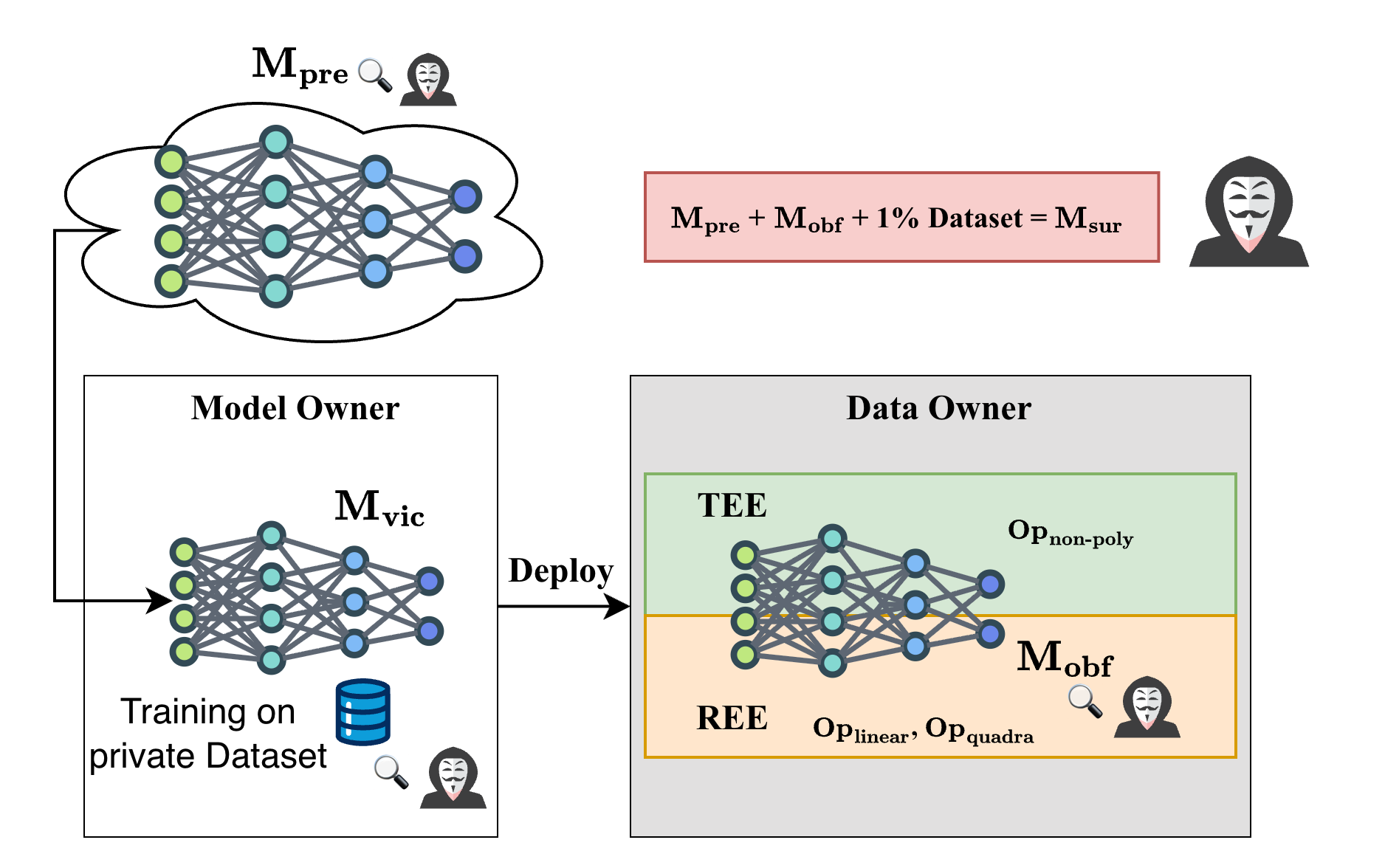}
  \Description{A threat-model illustration for obfuscation-based TEE--REE inference. The model owner starts from a public pretrained model M_pre, fine-tunes it on private data to obtain a victim model M_vic, and deploys an obfuscated runtime view M_obf in which the TEE keeps non-linear operators while the REE executes linear and quadratic operators. The attacker controlling the device combines M_pre, M_obf, and one percent of task data to construct a surrogate model M_sur.}
  \caption{Threat model. The model owner fine-tunes $\mathcal{M}_\textsf{pre}$ into a victim $\mathcal{M}_\textsf{vic}$ and deploys it as $\mathcal{M}_\textsf{obf}$; the device-owner attacker combines $\mathcal{M}_\textsf{pre}$, $T$ distinct exposed views, and $\le 1\%$ task data into a surrogate $\mathcal{M}_\textsf{sur}$.}
  \label{fig:bg_attack_surface}
\end{figure}

\noindent\textbf{System model.}
The model owner starts from a public pretrained model $\mathcal{M}_\textsf{pre}$ and fine-tunes it on private data to obtain the proprietary victim model $\mathcal{M}_\textsf{vic}$. $\mathcal{M}_\textsf{vic}$ is deployed on the user's device under the TEE--REE pipeline of \S~\ref{sec:bg_tee_ree}, producing the obfuscated runtime view $\mathcal{M}_\textsf{obf}$ to the REE.

\noindent\textbf{Defender goal and mechanism.}
The defender keeps $\mathcal{M}_\textsf{vic}$ usable while preventing recovery of $\mathbf{W}$. Following standard practice~\cite{wang2025arrows,xiong2025loro,zhou2023nnsplitter,sun2025tsqp,li2024translinkguard,shen2022soter}, the TEE primarily adopts a \emph{dynamic-key} strategy, refreshing obfuscation keys every few rounds to prevent attacks that accumulate observations under a static key.

\noindent\textbf{TEE--REE trust model.}
Following prior TSLP defenses~\cite{wang2025arrows,xiong2025loro,zhou2023nnsplitter,hou2022model,li2024translinkguard}, the GPU is an untrusted REE accelerator and the TEE (e.g., SGX, TrustZone, or equivalent) is the trust anchor. In each LLM inference, the REE sees the released input, the obfuscated weight $\widetilde{\mathbf{W}}_t$, and the product it computes; $\mathbf{W}$ and the post-recovery $\mathbf{X}\mathbf{W}$ remain inside TEE. 
REEs that support confidential computing (such as the NVIDIA Hopper and Blackwell GPUs) are out of scope.

\noindent\textbf{Attack goal.}
Instead of an exact white-box recovery of every parameter of $\mathcal{M}_\textsf{vic}$, the attacker aims to construct a high-utility surrogate model $\mathcal{M}_\textsf{sur}$ whose performance is comparable with $\mathcal{M}_\textsf{vic}$.

\noindent\textbf{Attacker capabilities.}
The attacker fully controls the device and can performs the following measurements.
\begin{itemize}[leftmargin=2em]
  \item \textbf{Boundary visibility.} The attacker records every tensor that crosses the TEE--REE boundary and every REE-side tensor derived from it. Across key refreshes, the attacker obtains $T$ distinct exposed views $\{\mathbf{W}^{e}_t\}_{t=1}^{T}$, where $\mathbf{W}^{e}_t \equiv \widetilde{\mathbf{W}}_t$.
  \item \textbf{Limited task data.} The attacker queries the deployed model to collect a limited amount of labeled task data---typically less than $1\%$ of the original training set---for fine-tuning the surrogate model~\cite{wang2025arrows}.
  \item \textbf{Public pretrained prior.} The attacker uses standard model-fingerprinting techniques~\cite{zeng2024huref,pasquini2025llmmap} to identify the public backbone $\mathcal{M}_\textsf{pre}$ and treats it as prior knowledge about $\mathcal{M}_\textsf{vic}$.
\end{itemize}

\noindent\textbf{Out-of-scope assumptions.}
Hardware side-channel and tran\-sient-execution attacks against the TEE are outside the scope of this work.

\section{Systematizing Linear Obfuscation}
\label{sec:primitives}

\subsection{Obfuscation Primitives} \label{sec:def_primitives}
The fundamental building blocks of existing TSLP methods are what we term \textit{Obfuscation Primitives}, which abstract the pair of operations used to mask a weight matrix before offloading computation to the REE and to recover the correct result inside the TEE. A trivial instantiation is an i.i.d.\ Gaussian 
mask $R$, with $\mathcal{O}(W) = W + R$ and 
$\mathcal{R}(X\widetilde{W}) = X\widetilde{W} - XR$. Although this 
construction offers both exact correctness and strong security, it is 
self-defeating in our setting: since $R$ is dense, computing $XR$ 
inside the TEE incurs a cost equivalent to the original operation 
$XW$, rendering the offloading to the REE fundamentally redundant. 
A viable primitive must therefore \emph{jointly} satisfy correctness, 
efficiency, and a meaningful bound on leakage, as formalized below.

\begin{definition}[Obfuscation Primitive]
\label{def:obfuscation-primitive}\label{def:primitive}
An \emph{obfuscation primitive} is a pair of (possibly randomized) 
operators $(\mathcal{O}, \mathcal{R})$, where $\mathcal{O}$ maps a 
weight matrix $W$ to an obfuscated matrix 
$\widetilde{W} = \mathcal{O}(W)$, and $\mathcal{R}$ reconstructs 
the intended product from $X\widetilde{W}$. The pair 
$(\mathcal{O}, \mathcal{R})$ is called a \emph{valid} obfuscation 
primitive if it satisfies the following three properties:
\begin{enumerate}[leftmargin=2em]
    \item \textbf{Correctness.} For every input matrix $X$ and weight 
    matrix $W$,
    \[
        \mathcal{R}_{W}\!\bigl(X \cdot \mathcal{O}_{W}(W)\bigr) \;=\; XW,
    \]
    i.e., the recovery operator restores exact functional equivalence 
    with the original linear computation.

    \item \textbf{Efficiency.} Let $C(\cdot)$ denote the computational 
    cost of an operation. The TEE-side cost of obfuscation and 
    recovery is negligible relative to the original matrix multiplication:
    \[
        C(\mathcal{O}) + C(\mathcal{R}) \;\ll\; C(XW).
    \]

    \item \textbf{Model-Extraction Resistance.}
      We characterize the security property of an obfuscation primitive as its
    resistance to model extraction, using a security game based on the threat
    model in \S~\ref{sec:threat_model}.
    \emph{Game overview.} With an observation budget $T$ and a labeled task-data
    budget $B$, the adversary attempts to extract the victim model.
    \emph{Winning condition.} The adversary uses any polynomial-time
    attack strategy 
    to maximize the task accuracy of an output surrogate $\widehat{\mathcal{M}}$, and wins if
    \[
      \operatorname{Acc}_{\mathcal{D}_{\mathrm{eval}}}(\widehat{\mathcal{M}})
      \ge
      \operatorname{Acc}_{\mathcal{D}_{\mathrm{eval}}}(\mathcal{M}_{\mathrm{vic}})-\delta.
    \]
    Here, $\operatorname{Acc}_{\mathcal{D}_{\mathrm{eval}}}(\mathcal{M})$
    denotes the accuracy of $\mathcal{M}$ on
    the evaluation set $\mathcal{D}_{\mathrm{eval}}$, and $\delta\ge 0$
    is the tolerated accuracy deficiency of the surrogate model relative to the victim model.\footnote{For the classification tasks evaluated in this paper, we adopt $\delta=0.05$. Thus, an attack is considered successful if the surrogate model's accuracy is no more than five percentage points below that of the victim model.} The formal game and its success criterion are given in
    Appendix~\ref{app:extraction-game}.
 \end{enumerate}
\end{definition}

\subsection{From Primitives to Schemes}
\label{sec:composability}

A central property of obfuscation primitives is
that primitives compose without degrading any of the three
properties: correctness remains exact, TEE-side computation cost grows only additively, and under the conditions stated below, the composition is
no easier to extract than individual components.

\begin{theorem}[Primitives are closed under composition]
\label{thm:composability}\label{thm:composition}
Let $\mathcal{P} = (\mathcal{O}_1, \mathcal{R}_1)$ and
$\mathcal{Q} = (\mathcal{O}_2, \mathcal{R}_2)$ be compatible
obfuscation primitives. Their \emph{composition}
$\mathcal{Q} \circ \mathcal{P}$ is the pair
$(\mathcal{O}_2 \circ \mathcal{O}_1,\;
\mathcal{R}_1 \circ \mathcal{R}_2)$, where obfuscations are applied
in order and recoveries in reverse so that each $\mathcal{R}_i$
undoes its corresponding $\mathcal{O}_i$. Under the conditions below,
this composed pair is itself a valid obfuscation primitive, with:
\begin{enumerate}
  \item[\textnormal{(i)}] \textbf{Exact correctness}, by chaining the
    correctness of $\mathcal{P}$ and $\mathcal{Q}$;
  \item[\textnormal{(ii)}] \textbf{Additive TEE-side cost},
    $C(\mathcal{Q} \circ \mathcal{P}) =
    C(\mathcal{P}) + C(\mathcal{Q})$, still negligible relative to
    $C(XW)$;
  \item[\textnormal{(iii)}] \textbf{No weaker model-extraction
    resistance}. Intuitively, an attack that extracts the composition
    would also induce an attack on each component. Under the formal
    full-transcript simulation conditions detailed in 
    Appendix~\ref{app:composition-proof}, the composition is therefore
    no easier to extract than either $\mathcal{P}$ or $\mathcal{Q}$.
\end{enumerate}
\end{theorem}

We provide the detailed proof in Appendix~\ref{app:composition-proof}.

Given this theorem, we define
the \emph{efficiency-preserving obfuscation family} generated by a \emph{base set} of obfuscation primitives $\mathcal{B} = \{\mathcal{P}_1, \dots, \mathcal{P}_N\}$ as
\[
    \mathcal{F}(\mathcal{B}) \;:=\;
    \Bigl\{\, \mathcal P =
    \mathcal{P}_{i_k} \circ \cdots \circ \mathcal{P}_{i_1}
    \;\Bigm|\;
    \substack{k \ge 1,\; i_1, \dots, i_k \in \{1, \dots, N\} \\
     C_{\mathrm{TEE}}(\mathcal P) \ll C(XW)}
    \,\Bigr\}.
\]
Here $C_{\mathrm{TEE}}(\mathcal P)$ denotes the accumulated
TEE-side obfuscation and recovery cost of the composition. Accordingly,
heavy compositions that exceed the efficiency budget in
Definition~\ref{def:obfuscation-primitive} are excluded from
$\mathcal{F}(\mathcal{B})$.{} In the next 
subsection, we decompose existing efficient TSLP schemes into their 
constituent primitives.

\subsection{Prior Schemes as Compositions of Primitives}
\label{sec:prior_works}

We now apply the framework of \S~\ref{sec:composability} to the
existing literature. Although prior TSLP schemes
\cite{hou2022model, li2024translinkguard, shen2022soter, sun2025tsqp, 
wang2025arrows, xiong2025loro, zhou2023nnsplitter} have been 
introduced as standalone designs, each with its own obfuscation 
scheme and security argument, we find that their matrix-level weight
transformations are built from a small set of primitives. This
subsection introduces a base set 
$\mathcal{B}_{\text{prior}} = \{\mathcal{P}_S, \mathcal{P}_L,
\mathcal{P}_\Pi, \mathcal{P}_D\}$ of four such primitives, and shows
that the matrix-level transformations summarized in
Table~\ref{tab:related_works} all belong to the family
$\mathcal{F}(\mathcal{B}_{\text{prior}})$.

\noindent\textbf{Taxonomy of Primitives.} 
The primitives fall into two paradigms based on their algebraic 
nature: \emph{additive} primitives, which mask $W$ by adding a 
structured noise term, and \emph{multiplicative} primitives, which 
transform $W$ by an invertible linear map. Their obfuscation and 
recovery operators are summarized in Table~\ref{tab:obf_comparison}.

\begin{table*}[t]
  \caption{The four obfuscation primitives identified in the 
  literature. $S$, $L = AB^\top$, $\Pi$, and $D$ denote a sparse, 
  low-rank, permutation, and non-singular diagonal matrix, 
  respectively. TEE cost is given for 
  $X \in \mathbb{R}^{m \times n}$, $W \in \mathbb{R}^{n \times n}$. 
  Security analysis is deferred to \S~\ref{sec:attack_insight}.}
  \label{tab:obf_comparison}
  \centering
  \small
  \setlength{\tabcolsep}{6pt}
  \begin{tabular}{@{}llcccc@{}}
    \toprule
    \textbf{Paradigm} & \textbf{Primitive} 
      & $\mathcal{O}(W)$ 
      & $\mathcal{R}(X\widetilde W)$ 
      & \textbf{TEE cost} 
      & \textbf{Notes} \\
    \midrule
    \multirow{2}{*}{Additive} 
      & $\mathcal{P}_S$ \;sparse mask 
      & $W + S$ 
      & $X\widetilde W - XS$ 
      & $O(m \cdot \nnz)$
      & $\nnz \ll n^2$ \\
      & $\mathcal{P}_L$ \;low-rank mask 
      & $W + AB^\top$ 
      & $X\widetilde W - (XA)B^\top$ 
      & $O(mn)$ 
      & $L = AB^\top$, $r$ small \\
    \midrule
    \multirow{2}{*}{Multiplicative} 
      & $\mathcal{P}_\Pi$ \;col-wise perm. 
      & $W\Pi$ 
      & $X\widetilde W \, \Pi^{-1}$ 
      & memcpy 
      & $\Pi\Pi^\top = I$ \\
      & $\mathcal{P}_D$ \;col-wise scale 
      & $WD$ 
      & $X\widetilde W \, D^{-1}$ 
      & $O(mn)$ 
      & $d_i \neq 0$ \\
    \bottomrule
  \end{tabular}
\end{table*}

\noindent\textbf{Additive Primitives.}
Both additive primitives share the form 
$\mathcal{O}(W) = W + R$ and 
$\mathcal{R}(X\widetilde W) = X\widetilde W - XR$, 
differing only in the structure of the additive term $R$. Their efficiency comes from the structure of $R$, which makes 
$XR$ cheaper to compute than $XW$ itself.

A \emph{sparse mask} sets $R = S$, where $S$ is a sparse matrix 
with $\nnz \ll n^2$ non-zero entries sampled by the TEE.\footnote{$\nnz$ denotes the number of nonzero entries in a sparse mask.} Recovery 
amounts to subtracting $XS$ from $X\widetilde W$, where $XS$ is a 
sparse--dense multiplication of cost $O(m \cdot \nnz)$, far below 
the $O(mn^2)$ cost of $XW$.

For a \emph{low-rank mask}, $R = AB^\top$ with 
$A, B \in \mathbb{R}^{n \times r}$ sampled by the TEE at small 
rank $r$. The TEE recovers $XW$ by subtracting $(XA)B^\top$ from 
$X\widetilde W$, two thin matrix multiplications of total cost 
$O(mn)$, again far below $O(mn^2)$.

\noindent\textbf{Multiplicative Primitives.}
 We use \emph{one-sided permutation and scaling} to refer to transformations
acting on either rows or columns. Since the two orientations are
equivalent up to transposition, e.g., $(AW)^\top=W^\top A^\top$, we use
the column-wise convention below for simplicity.
 Both multiplicative primitives share the form 
$\mathcal{O}(W) = WT$ and 
$\mathcal{R}(X\widetilde W) = X\widetilde W \, T^{-1}$, where $T$ 
is an invertible matrix sampled by the TEE. Their efficiency comes from the structure of $T^{-1}$, which 
reduces TEE-side recovery to operations far cheaper than a dense 
matrix multiplication.

A \emph{column-wise permutation} sets $T = \Pi$, an $n \times n$ 
permutation matrix sampled by the TEE. The recovery 
$X\widetilde W \, \Pi^{-1}$ reduces to permuting the columns of 
$X\widetilde W$ back to their original order, which the TEE 
performs as a memory copy with no arithmetic.

A \emph{column-wise scaling} sets $T = D = \mathrm{diag}(d_1, 
\dots, d_n)$, a non-singular diagonal matrix sampled by the TEE. 
The recovery $X\widetilde W \, D^{-1}$ reduces to dividing the 
$i$-th column of $X\widetilde W$ by $d_i$, which the TEE performs as $mn$ scalar divisions, at total cost of $O(mn)$.

\noindent\textbf{Existing Schemes as Compositions.}
Table~\ref{tab:related_works} maps the matrix-level weight transformations
of prior TSLP schemes onto these four primitives. Defining
\[
    \mathcal{B}_{\text{prior}} 
    \;:=\; \bigl\{\, \mathcal{P}_S,\;\;
                    \mathcal{P}_L,\;\;
                    \mathcal{P}_\Pi,\;\;
                    \mathcal{P}_D \,\bigr\},
\]
we observe that every weight transformation listed in
Table~\ref{tab:related_works} is an element of
$\mathcal{F}(\mathcal{B}_{\text{prior}})$ --- whether built on a single
primitive or, as in the case of ArrowCloak~\cite{wang2025arrows}, composed of several.

\begin{table*}[t]
  \caption{Representative TSLP schemes mapped onto the four 
  primitives of $\mathcal{B}_{\text{prior}}$. Six of the seven 
  schemes use a single primitive; only ArrowCloak composes 
  multiple primitives in one scheme. The \textbf{Mode} column 
  indicates whether the same weight matrix is obfuscated by a 
  single fixed mask (static) or by multiple masks across inferences 
  (dynamic).}
  \label{tab:related_works}
  \centering
  \small
  \setlength{\tabcolsep}{5pt}
  \begin{tabular}{@{}llcccccccc@{}}
    \toprule
    \textbf{Scheme} & \textbf{Venue} 
      & $\mathcal{P}_S$ & $\mathcal{P}_L$ 
      & $\mathcal{P}_\Pi$ & $\mathcal{P}_D$ 
      & \textbf{Mode}
      &  \textbf{Original transformation}
      & \textbf{$\mathcal{O}(W)$} 
      & \textbf{$\mathcal{R}(X\widetilde W)$} \\
    \midrule
    NNSplitter \cite{zhou2023nnsplitter} & ICML'23 
      & \ding{51} & & & 
      & static
      & $\widetilde W=W+S,\ \mathrm{nnz}(S)\ll n^2$
      & $W + S$ & $X\widetilde W - XS$ \\
    Magnitude \cite{hou2022model} & TDSC'21 
      & \ding{51} & & & 
      & static
      & $\widetilde W=W+S,\ \mathrm{nnz}(S)\approx0.01n^2$
      & $W + S$ & $X\widetilde W - XS$ \\
    LoRO \cite{xiong2025loro} & NeurIPS'25 
      & & \ding{51} & & 
      & dynamic
      & $\widetilde W_\ell=W_\ell+D_\ell^{\mathrm{LoRO}},\ \operatorname{rank}(D_\ell^{\mathrm{LoRO}})\ll n$
      & $W + AB^\top$ & $X\widetilde W - (XA)B^\top$ \\
    TransLinkGuard \cite{li2024translinkguard} & MM'24 
      & & & \ding{51} & 
      & static
      & $\widetilde W_{q,k,v}=\Pi_i^\top W_{q,k,v},\ \widetilde W_o=W_o\Pi_i$
      & $W\Pi$ & $X\widetilde W \Pi^{-1}$ \\
    Soter \cite{shen2022soter} & ATC'22 
      & & & & \ding{51} 
      & dynamic
      & $\widetilde W_i=\mu_iW_i$
      & $WD$ & $X\widetilde W D^{-1}$ \\
    TSQP \cite{sun2025tsqp} & S\&P'25 
      & & & & \ding{51} 
      & static
      & $\widetilde W=W_{\mathrm{dissim}}/s$
      & $WD$ & $X\widetilde W D^{-1}$ \\
    ArrowCloak \cite{wang2025arrows} & Security'25 
      & & \ding{51} & \ding{51} & \ding{51} 
      & dynamic
      & $\widetilde W=(WD_1+\mathbf{v}\mathbf{1}_n^\top D_2)\Pi$
      &  $(W+L)D\Pi$
      &  $X\widetilde W\Pi^{-1}D^{-1}-XL$ \\
    \bottomrule
  \end{tabular}
\end{table*}

Table~\ref{tab:related_works} gives a scheme-by-scheme correspondence between
each published weight transformation and its representation using our primitives.
The last two columns show the equivalent obfuscation and TEE-side recovery.
For TransLinkGuard~\cite{li2024translinkguard}, adjacent layers use different permutations, so the TEE
rearranges the output of one layer into the order expected by the next. For ArrowCloak~\cite{wang2025arrows},
setting $D=D_1$ and
$L=\mathbf{v}\mathbf{1}_n^\top D_2D_1^{-1}$ gives
$(WD_1+\mathbf{v}\mathbf{1}_n^\top D_2)\Pi=(W+L)D\Pi$.
Our abstraction captures the core weight-obfuscation and recovery algorithms
at the TEE--REE boundary, which are central to protecting offloaded weights.

 \subsection{Security Boundary of \texorpdfstring{$\mathcal{B}_{\text{prior}}$}{B\_prior}}
\label{sec:security_boundary}

\S~\ref{sec:prior_works} established that the matrix-level transformations listed in Table~\ref{tab:related_works} lie in
 $\mathcal{F}(\mathcal{B}_{\text{prior}})$. This motivates the following
question: does this family admit a ``most secure'' composition?

\begin{definition}[Security Boundary]
\label{def:security-boundary}
Let $\mathcal{B}$ be a base set of obfuscation primitives. A
scheme $\mathcal{O}^*$ is the \emph{security boundary} of
$\mathcal{B}$ if it satisfies the efficiency requirement in
Definition~\ref{def:obfuscation-primitive} and every
efficiency-preserving $\mathcal{P} \in \mathcal{F}(\mathcal{B})$
  can be realized by $\mathcal{O}^*$ under some choice of parameters.
 Therefore, $\mathcal{O}^*$ can be viewed as the ``most secure'' scheme
attainable from $\mathcal{B}$.
\end{definition}

\begin{theorem}[Security boundary of $\mathcal{B}_{\text{prior}}$]
\label{thm:omax}
The security boundary of $\mathcal{B}_{\text{prior}}$ is
\begin{equation}
    \mathcal{O}_{\text{prior}}(W;\; S, L, D, \Pi) 
    \;=\; (W + S + L)\, D\, \Pi,
    \label{eq:omax}
\end{equation}
where $S$, $L$, $D$, and $\Pi$ represents sparse, low-rank,
non-singular diagonal, and permutation matrices, respectively.
That is, every efficiency-preserving composition $\mathcal{P} \in 
\mathcal{F}(\mathcal{B}_{\text{prior}})$ is an instance of
$\mathcal{O}_{\text{prior}}$ for some $S$, $L$, $D$, and $\Pi$.
\end{theorem}

\begin{proof}
We show that any efficiency-preserving composition $\mathcal{P}$ of primitives in 
$\mathcal{B}_{\text{prior}}$ can be rewritten into the form of 
Eq.~\eqref{eq:omax} via a sequence of algebraic identities.
Represent $\mathcal{P}$ as a sequence of primitives drawn
from $\{S, L, D, \Pi\}$ and applied to $W$. The following rewriting rules reduce any such 
word to canonical form:
\begin{align}
  D_1 D_2 &\;\longrightarrow\; D_{12}
  &&\text{diagonal closed under product}
  \label{eq:rw-dd}\\
  \Pi_1 \Pi_2 &\;\longrightarrow\; \Pi_{12}
  &&\text{perm.\ closed under product}
  \label{eq:rw-pp}\\
  \Pi\, D &\;\longrightarrow\; D'\, \Pi
  &&D' = \Pi D \Pi^{-1} \text{ diagonal}
  \label{eq:rw-pid}\\
  (W + L)\,M &\;\longrightarrow\; W M + L'
  &&M \in \{D, \Pi\},\; L' \text{ stays low-rank}
  \label{eq:rw-lm}\\
  (W + S)\,M &\;\longrightarrow\; W M + S'
  &&M \in \{D, \Pi\},\; S' \text{ stays sparse}
  \label{eq:rw-sm}\\
  L_1 + L_2 &\;\longrightarrow\; L_{12}
  &&\mathrm{rank}\,L_{12} \leq \mathrm{rank}\,L_1 + \mathrm{rank}\,L_2
  \label{eq:rw-ll}\\
  S_1 + S_2 &\;\longrightarrow\; S_{12}
  &&\nnz S_{12} \leq \nnz S_1 + \nnz S_2
  \label{eq:rw-ss}
\end{align}
We apply these rules in three phases.
First, repeated application of Rules~\eqref{eq:rw-lm} 
and~\eqref{eq:rw-sm} distributes every multiplicative primitive 
through the enclosing additive structure, yielding
\[
    W\, M_1 M_2 \cdots M_k 
    \;+\; \sum_i L_i \;+\; \sum_j S_j,
\]
where each $M_t \in \{D, \Pi\}$.
Second, Rules~\eqref{eq:rw-dd}, \eqref{eq:rw-pp}, 
and~\eqref{eq:rw-pid} reduce the multiplicative chain 
$M_1 M_2 \cdots M_k$ to $D^* \Pi^*$ (Rule~\eqref{eq:rw-pid} pushes 
permutations rightward; Rules~\eqref{eq:rw-dd} 
and~\eqref{eq:rw-pp} fuse adjacent diagonals and adjacent 
permutations).
Finally, Rules~\eqref{eq:rw-ll} and~\eqref{eq:rw-ss} fuse the 
additive masks of the same type, giving
\[
    \mathcal{P}(W) \;=\; W D^* \Pi^* + L^* + S^*.
\]
Setting $\widehat{L} = L^* (D^* \Pi^*)^{-1}$ and 
$\widehat{S} = S^* (D^* \Pi^*)^{-1}$, both still low-rank and 
sparse respectively, we obtain
\[
    \mathcal{P}(W) \;=\; (W + \widehat{S} + \widehat{L})\, D^* \Pi^* 
    \;=\; \mathcal{O}_{\text{prior}}(W;\, \widehat{S}, \widehat{L}, 
    D^*, \Pi^*),
\]
which is the canonical form of Eq.~\eqref{eq:omax}.
\end{proof}

Every efficiency-preserving composition built from
$\mathcal{B}_{\text{prior}}$ is an instance of $\mathcal{O}_{\text{prior}}$.
Whether $\mathcal{O}_{\text{prior}}$ offers adequate model-extraction protection is the question we take up in \S~\ref{sec:attack}.

 \section{Attack Intuition}
\label{sec:attack-insight}\label{sec:attack_insight}\label{sec:attack_overview}\label{sec:leaks}

Before presenting \sysattack in detail (\S~\ref{sec:attack}), we first
explain \emph{why} $\mathcal{O}_\textsf{prior}$ is vulnerable. The
high-level intuition is that no matter how the four primitives
$\{\Pi, D, S, L\}$ are composed, the \emph{column directions} of the
victim weight $W_\textsf{vic}$ remain largely exposed. We articulate
this as three insights, each pointing to a structural limitation of one
or more primitives.

\begin{insightbox}{Insight 1 (Column atomicity).}
Each victim column survives in $\widetilde{W}$ as a transformed copy of
a single original column. $\Pi$ and $D$ relabel and rescale columns but
never mix them.
\end{insightbox}

\noindent
$\Pi$ permutes columns and $D$ scales them entry-wise per column;
neither operation forms linear combinations \emph{across} columns.
Hence, the set of column directions of $W_\textsf{vic}$ survives
the perturbations induced by $\Pi$ and $D$.

\begin{insightbox}{Insight 2 (Sparse coverage gap).}
A vanishing fraction of entries are perturbed; the rest leak in
cleartext. Sparse mask $S$ has $\mathrm{nnz} \ll n^{2}$ nonzeros.
\end{insightbox}

\noindent
Practical deployments keep $S$ extremely sparse for efficiency, so
nearly every entry of each column of $W_\textsf{vic}$ is released
verbatim. A column direction is a global property of all $n$ entries,
and corrupting only a small fraction of them barely tilts it.

\begin{insightbox}{Insight 3 (Low-rank coverage gap).}
Each column's $(n{-}r)$-dim orthogonal component leaks in cleartext.
Low-rank mask $L = AB^{\top}$ confines noise to an $r$-dim subspace.
\end{insightbox}

\noindent
Because $L$ has rank at most $r \ll n$, every column of $L$ lives in a
single $r$-dimensional subspace. Projecting away from this subspace
removes $L$ entirely and exposes the remaining $(n{-}r)$ dimensions of
each victim column without any noise at all.

\parab{Why direction leakage suffices.}
$\mathcal{O}_\textsf{prior}$ fails to hide the column directions of
$W_\textsf{vic}$. Since the backbone architecture is public, building
$\mathcal{M}_\textsf{sur}$ reduces to recovering the weight matrices. Once the
column directions are obtained, the attacker only needs to fix them
and fine-tune the column magnitudes to build $\mathcal{M}_\textsf{sur}$.
\S~\ref{sec:collapse} develops the attack.
\section{The \sysattack Pipeline}
\label{sec:collapse}\label{sec:attack}\label{sec:our-attack}

\begin{figure*}[!t]
    \centering
    \includegraphics[width=0.95\textwidth]{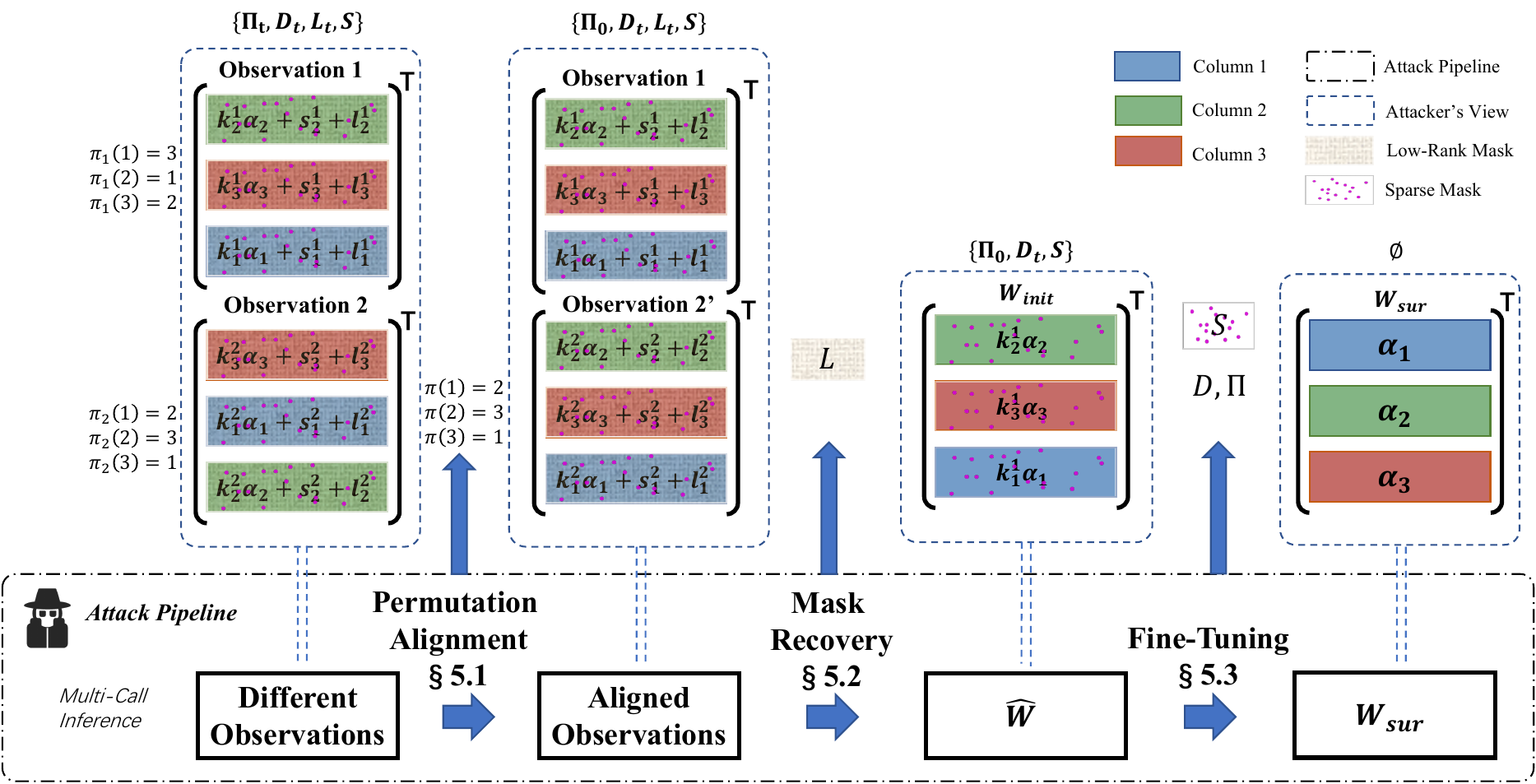}
    \Description{Three-stage pipeline diagram. Multiple obfuscated observations enter Stage 1, which computes dominant-subspace projections and recovers a per-round permutation to produce column-aligned views; these feed Stage 2, which identifies paired mask subspaces through intersections of column spans and jointly solves for relative column scales and mask components to subtract the low-rank mask; the output feeds Stage 3, which matches columns against the public prior, flags sparse support via residual voting, and calibrates per-column lengths before returning the surrogate weight.}
    \caption{\sysattack against \priorboundary: align columns (Stage~1), separate the low-rank mask (Stage~2), fine-tuning the length of columns and recover sparse mask (Stage~3).}
    \label{fig:pipeline}
\end{figure*}

\S~\ref{sec:leaks} has shown that $\mathcal{O}_{\text{prior}}$ leaks
$W$ along three structural channels: column directions survives the perturbations by $\Pi$ and $D$, most entries survives $S$, and most directions survives
 $L$. This section turns those leaks into an executable
model-extraction attack, which we call \emph{Collapse}.
Specifically, the attacker observes $T$ obfuscated views of a single weight matrix,
\begin{equation}
\widetilde{W}_t \;=\; (W + S + L_t)\,D_t\,\Pi_t \;\in\; \mathbb{R}^{m\times n},
\quad t = 1,\ldots,T,
\label{eq:view}
\end{equation}
with a static $S$ and dynamically-refreshed $L_t$, $D_t$, and $\Pi_t$. The attacker additionally has
access to a public pretrained weight $W_{\mathrm{pre}}\in\mathbb{R}^{m\times n}$
identified by model fingerprinting, and to a labeled task-data budget
$\mathcal{D}_{\mathrm{task}}$ of size at most $1\%$ of the original
training set. The output is a surrogate model $\mathcal{M}_\textsf{sur}$ whose
downstream accuracy approximates the victim $\mathcal{M}_\textsf{vic}$.

Figure~\ref{fig:pipeline} summarizes the pipeline. \emph{Stage~1}
(\S~\ref{sec:stage1}) consumes view-to-view consistency to discharge
$\{\Pi_t\}$, returning $K\geq 4$ views in a common column ordering.
 Stage~1 treats $W+S$ as a whole and uses cross-view consistency to align the relative permutations $\{\Pi_t\}$.\ 
\emph{Stage~2} (\S~\ref{sec:stage2}) uses mask-subspace intersections from four aligned views to
discharge $\{L_t\}$, producing an anchor-ordered estimate $\widehat{W+S}$. \emph{Stage~3} (\S~\ref{sec:stage3}) consumes
the prior $W_{\mathrm{pre}}$ and the task-data budget to discharge the
column scaling $\{D_t\}$, the unknown column ordering, and any
residual error, returning $\mathcal{M}_\textsf{sur}$. Stages~1 and~2 run independently per weight matrix; Stage~3 runs
once on the assembled model.

\subsection{Stage 1: Cross-View Column Alignment}
\label{sec:stage1}

Given the $T$ raw views, we extract $K\geq 4$ that share a common column
ordering. The recovered ordering is \emph{relative}: for every accepted
pair $(a,b)$ we obtain $\Pi_a \Pi_b^{-1}$ exactly, while the absolute
$\Pi_t$ remain unknown.  We treat $W+S$ as a whole and rely on the linear dependency detection (LDD) oracle of \S~\ref{sec:ldd}, an algebraic test that decides whether two column subsets of $\widetilde{W}_a$ and
$\widetilde{W}_b$ correspond to the same victim columns.
We then call this oracle repeatedly into an
explicit recovery of $\Pi_a \Pi_b^{-1}$.

\subsubsection{The LDD Oracle}
\label{sec:ldd}

For column subsets $I_a, I_b \subset [n]$ with $|I_a| = |I_b| = k$, stack
the corresponding columns into
\[
C \;=\; \bigl[\,\widetilde{W}_a[:, I_a] \;\big|\; \widetilde{W}_b[:, I_b]\,\bigr] \;\in\; \mathbb{R}^{m\times 2k},
\]
and let $J_a = \Pi_a^{-1}(I_a)$, $J_b = \Pi_b^{-1}(I_b)$ be the
corresponding indices in the shared matrix $W+S$. The dimension of $C$'s column space
is determined by to what extent $J_a$ and $J_b$ overlap: 

\begin{itemize}
\item \textbf{Full overlap} ($J_a = J_b$). Both blocks span the same $k$
victim directions, and $L_a, L_b$ each contribute a rank-$r$ subspace
generic with respect to $\colop(W+S)$. The columns of $C$ admit
\[
N_{\mathrm{full}} \;=\; k + 2r
\]

linearly independent vectors among them.
\item \textbf{Partial overlap} ($|J_a\cup J_b| = k+s$, $s\geq 1$). The
victim-direction count rises to $k+s$, and $L_a, L_b$ contribute a
further $2r$, so the columns admit
\[
N_{\mathrm{part}} \;=\; k + s + 2r
\]
linearly independent vectors.
\end{itemize}

\parab{Choice of $k$.} 
We use the rank $\rho$ of $C$ as our estimate of the number of
linearly independent columns it contains. The parameter $k$ then
controls how large an overlap deficit $s$ we can detect: because $C$
has only $2k$ columns, its rank is bounded by $\rho \leq 2k$, so the
identity $\rho = N_{\mathrm{part}} = k + s + 2r$ holds only while
$k + s + 2r \leq 2k$, i.e.\ $s \leq k - 2r$. Beyond this threshold, $\rho$ saturates at $2k$, making different
values of $s$ indistinguishable. For instance, taking $k = 32$ when $r = 2$ yields
a detection range of $s \leq 28$, so every overlap situation with
$|J_a \cap J_b| \in \{4, 5, \dots, 32\}$ is resolved exactly from
$\rho$. Thus even a small overlap is immediately visible in $\rho$.

\begin{algorithm}[t]
\caption{LDD Oracle}
\label{alg:ldd}

\KwIn{$A, B \in \mathbb{R}^{m\times k}$, rank threshold $\tau$}
\KwOut{Numerical rank $\rho$}

$C \leftarrow [\,A \mid B\,]$\;

$\rho \leftarrow \textsc{Rank}(C,\, \tau)$\;
\Return $\rho$\;

\end{algorithm}

\subsubsection{Recovering Relative Permutations}
\label{sec:permrec}

The LDD oracle returns only the size of the overlap between two
index sets $J_a$ and $J_b$, without establishing a one-to-one
correspondence between them. In this subsection, we build
an alignment algorithm on top of the oracle that recovers the full
relative permutation $\Pi_a \Pi_b^{-1}$ between two views. The
algorithm proceeds in three phases. We first
\emph{seed} the alignment by locating a pair of size-$k$ column
subsets that the oracle certifies as fully overlapping. We then
\emph{disambiguate} the bijection inside this seed by removing one
column from each side at a time and querying the oracle. Finally we
\emph{propagate} the bijection to all remaining columns one at a
time.

\parab{Phase 1: seeding.}
A \emph{seed} is a pair of size-$k$ column subsets $I_a, I_b$ that
fully overlap, i.e., $|J_a \cap J_b| = k$. Our strategy is to first
find $I_a$ and $I_b$ that partially overlap, then grow this overlap
until it is complete. The first step is essentially free: even when
$I_a$ and $I_b$ are drawn uniformly at random, the LDD oracle
reports the overlap $|J_a \cap J_b|$ with high probability, giving
us a partial overlap to start from. Next, we grow the overlap by replacing one column of $I_b$ at a time with a new column from $[n] \setminus I_b$, keeping only those replacements that increase $|J_a \cap J_b|$. We continue this process until $|J_a \cap J_b| = k$, at which point $(I_a, I_b)$ is a seed.

\parab{Phase 2: disambiguating.} A seed gives us two fully overlapping sets, but we still need the element-level mapping between them. To recover it, we remove one column from each of $I_a$ and $I_b$ and ask the LDD oracle whether the remaining $k - 1$ columns on each side still fully overlap. If they do, the two removed columns must correspond to the same underlying column. Repeating this test over all candidate pairs gradually determines the full correspondence between $I_a$ and $I_b$.

\parab{Phase 3: propagating.}
Following the phase 1 and phase 2, a one-to-one correspondence has been established between the $k$ columns of a seed pair $(I_a, I_b)$, and it remains to extend this correspondence to all columns. A natural starting point is to match the columns one at a time. We append an unmatched column to each side and query the LDD oracle. If the two sides remain fully overlapping, the appended columns are matched. By repeating this procedure until no unmatched columns remain, the full correspondence is eventually recovered.

\begin{algorithm}[t]
\newcommand{\BlueAlgorithmComment}[1]{{\color{black}\texttt{#1}}}\SetCommentSty{BlueAlgorithmComment}
\caption{Permutation Alignment via the LDD Oracle}
\label{alg:permrec}
\KwIn{Views $\widetilde{W}_a, \widetilde{W}_b \in \mathbb{R}^{m\times n}$; subset size $k$}
\KwOut{Bijection $\pi : [n] \to [n]$ matching columns of $\widetilde{W}_b$ to those of $\widetilde{W}_a$}

\tcp{$\textnormal{\scshape Ovl}(A, B) :=$ number of shared victim columns between $A$ and $B$, read off the rank returned by the LDD oracle.}

\tcc{Phase 1 (seeding): grow a random pair until it fully overlaps.}
draw $I_a, I_b \subset [n]$ uniformly at random with $|I_a| = |I_b| = k$\;

\tcp{When $k^2/n$ comfortably exceeds $2r$, random blocks likely satisfy $\textnormal{\scshape Ovl}>2r$; otherwise resample.\protect\footnotemark}

\While{$\textnormal{\scshape Ovl}(\widetilde{W}_a[:, I_a],\, \widetilde{W}_b[:, I_b]) < k$}{
    pick $y \in I_b$ and $y' \in [n]\setminus I_b$ such that swapping $y \mapsto y'$ in $I_b$ strictly increases $\textnormal{\scshape Ovl}$; apply it\;
}
\tcp{$(I_a, I_b)$ now index the same $k$ victim columns, in unknown order.}
\BlankLine
\tcc{Phase 2 (disambiguating): match the $k$ seed columns one-to-one.}
$\pi \leftarrow \emptyset$;\ $U_b \leftarrow I_b$\;
\For{$x \in I_a$}{
    \For{$y \in U_b$}{
        \If{$\textnormal{\scshape Ovl}(\widetilde{W}_a[:, I_a\setminus\{x\}],\, \widetilde{W}_b[:, I_b\setminus\{y\}]) = k-1$}{
            $\pi(x) \leftarrow y$;\ $U_b \leftarrow U_b \setminus \{y\}$;\ \textbf{break}\;
            \tcp{removing $x$ on the left and $y$ on the right preserves full overlap iff $\pi(x)=y$.}
        }
    }
}
\tcp{$\pi$ now bijects $I_a$ to $I_b$ on the seed; the remaining $n-k$ columns are still unmatched.}
\BlankLine
\tcc{Phase 3 (propagating): extend $\pi$ to the remaining $n-k$ columns.}
$U_a \leftarrow [n] \setminus I_a$;\ $U_b \leftarrow [n] \setminus I_b$\;
\While{$U_a \neq \emptyset$}{
    pick any $x \in U_a$\;
    \For{$y \in U_b$}{
        \If{$\textnormal{\scshape Ovl}(\widetilde{W}_a[:, I_a \cup \{x\}],\, \widetilde{W}_b[:, I_b \cup \{y\}]) = k+1$}{
            $\pi(x) \leftarrow y$;\ $U_a \leftarrow U_a\setminus\{x\}$;\ $U_b \leftarrow U_b\setminus\{y\}$;\ \textbf{break}\;
        }
    }
}
\Return $\pi$\;
\end{algorithm}
\footnotetext{In the BERT-Base setting~\cite{devlin2019bert} ($n=m=768$ and $r=2$), choosing $k=80$ gives $\Pr[\textnormal{\scshape Ovl}>2r]\approx93.92\%$ for one random block pair.}
\parab{Algorithm}
Algorithm~\ref{alg:permrec} assembles the three phases into a single
procedure that recovers $\pi$ from two raw views. The cost is
dominated by Phase 3, which makes $O(n^2)$ LDD calls; Phase 1 and
Phase 2 together contribute only $O(k^2)$ calls.

\subsection{Stage 2: Mask Recovery}
\label{sec:stage2}\label{sec:attack_stage2}

We select four aligned views from Stage~1. Working in their common
column order and suppressing the shared permutation $\Pi_0$, we write
\begin{equation}
\widetilde{W}_t = (W + S + L_t)D_t,
\qquad
t = 0,1,2,3.
\label{eq:stage2-input}
\end{equation}
Let $H=W+S$ and $d_{tj}=D_t[j,j]\ne0$. Our goal is
$\hat W=HD_0$: Stage~2 removes the low-rank masks while
leaving the shared sparse mask and anchor-view scaling to Stage~3.

\parab{Step 1: Identify mask subspaces.}
Choose disjoint column blocks $I,J\subseteq[n]$, each of size $2r$.
For each view pair $(a,b)\in\{(0,1),(2,3)\}$, concatenate the
selected columns from both views within each block, then take the
intersection of the resulting column spaces:
\begin{equation}
\begin{aligned}
\mathcal V_{ab}(I)
&=\colop[\widetilde W_a[:,I],\widetilde W_b[:,I]],\\
\mathcal V_{ab}(J)
&=\colop[\widetilde W_a[:,J],\widetilde W_b[:,J]],\\
\mathcal U_{ab}
&=\mathcal V_{ab}(I)\cap\mathcal V_{ab}(J).
\end{aligned}
\label{eq:stage2-intersection}
\end{equation}
Here $\widetilde W_t[:,X]$ selects the columns indexed by $X$,
$[A,B]$ denotes horizontal concatenation, and $\colop(M)$ is the
space spanned by the columns of $M$.
The column spaces $\mathcal V_{ab}(I)$ and $\mathcal V_{ab}(J)$
share exactly the $2r$ independent mask directions spanned by $L_a$
and $L_b$. Their intersection therefore extracts
$\mathcal U_{ab}=\colop[L_a,L_b]$.\footnote{We assume that the
$2r$ basis directions of $L_a,L_b$ and the $4r$ selected columns of
\mbox{$H[:,I],H[:,J]$} are jointly linearly independent, and that each
difference block \mbox{$(L_a-L_b)[:,X]$} spans all $2r$ mask directions
for $X\in\{I,J\}$. In BERT-Base ($m=n=768$, $r=2$), these
conditions require only $6r=12$ independent directions in
$\mathbb R^{768}$.}
Having identified $\mathcal U_{01}$ and $\mathcal U_{23}$, we use
SVD to obtain their orthonormal bases
$B_{01},B_{23}\in\mathbb R^{m\times2r}$.

\parab{Step 2: Recover the shared columns.}
Define $P_{ab}=I_m-B_{ab}B_{ab}^{\top}$, the orthogonal projector
onto $\mathcal U_{ab}^{\perp}$. The key observation is that the
low-rank masks in each view pair affect only components within
$\mathcal U_{ab}$, leaving its orthogonal complement unaffected:
\[
P_{ab}\widetilde W_t=P_{ab}HD_t,\qquad t\in\{a,b\}.
\]
We have $\mathcal U_{01}^{\perp}+\mathcal U_{23}^{\perp}=\mathbb R^m$,
so the two mask-free projections jointly provide information spanning
the entire space.\footnote{For independent mask directions drawn from
full-dimensional continuous distributions in $\mathbb R^m$,
$\mathcal U_{01}\cap\mathcal U_{23}=\{0\}$ almost surely when
$4r\le m$.}
However, these projections retain different unknown
column scales, requiring joint recovery of the relative scales and
mask components. Specifically, for each column $j$, let
$y_{tj}=\widetilde W_t[:,j]$ and solve
\begin{equation}
\begin{bmatrix}B_{01}&-B_{23}&y_{2j}\end{bmatrix}
\begin{bmatrix}\alpha_j\\\beta_j\\\gamma_j\end{bmatrix}=y_{0j},
\label{eq:stage2-joint}
\end{equation}
where $\alpha_j,\beta_j\in\mathbb R^{2r}$ are mask coordinates and
$\gamma_j\in\mathbb R$ is the relative scale.
Solving the system recovers the mask component
$B_{01}\alpha_j=d_{0j}L_0[:,j]$.\footnote{When $4r\ll m$, the coefficient matrix $[B_{01},-B_{23},y_{2j}]$
has full column rank $4r+1$ with overwhelming probability.}
Subtracting the mask then yields
\begin{equation}
\hat W[:,j]=y_{0j}-B_{01}\alpha_j=H[:,j]d_{0j}.
\label{eq:stage2-output}
\end{equation}
The recovered $\hat W=(W+S)D_0$ is passed to Stage~3.

\subsection{Stage 3: Fine-tuning}
\label{sec::attack_stage3}\label{sec:stage3}

Stage~3 takes the anchor-ordered estimate $\hat{W}=(W+S)D_0$ and the public prior $W_{\mathrm{pre}}$ and turns it into a surrogate model $\mathcal{M}_\textsf{sur}$. To do so, we first align the recovered columns with $W_{\mathrm{pre}}$, then estimate a candidate sparse support, and finally calibrate the column lengths and fine-tune the resulting initialization.

\parab{Step 1: direction recovery.}
This step recovers the column permutation $\Pi_0$ by exploiting direction similarity between $W_{\mathrm{pre}}$ and $\hat{W}$~\cite{wang2025arrows}. We find an index mapping $\sigma(\cdot)$, the inverse of $\Pi_0$, that maps each column of $\hat{W}$ back to its position in the column order of $W_{\mathrm{pre}}$. We use cosine similarity to quantify directional similarity: at each step, the attacker fixes a column in $\hat{W}$ and, within $W_{\mathrm{pre}}$, selects the vector with the largest cosine similarity to it: 
\begin{equation}
\sigma(i) = \operatorname{argmax}_{j} \cos(\hat{w}^i, w^j_{\mathrm{pre}}).
\end{equation}
Since each column of $\hat{W}$ has only a small fraction of entries perturbed by $S$, the directional similarity between $\hat{W}$ and $W_{\mathrm{pre}}$ is preserved.

\parab{Step 2: sparse-support estimation.}
After aligning the recovered columns with $W_{\mathrm{pre}}$, we first estimate the remaining scale of each column by least-squares projection:
\begin{equation}
a_j
=
\frac{\left\langle \hat W[:,j],W_{\mathrm{pre}}[:,j]\right\rangle}
{\left\|W_{\mathrm{pre}}[:,j]\right\|_2^2}.
\end{equation}
Using the fitted scale, we measure the normalized residual of each entry as
\begin{equation}
R_{ij}
=
\frac{\left|\hat W_{ij}-a_jW_{\mathrm{pre},ij}\right|}
{\sqrt{u_i v_j}}.
\end{equation}
where $u_i$ and $v_j$ are the typical residual magnitudes of row $i$ and column $j$. We then select the candidate sparse support by
\begin{equation}
\hat\Omega=\operatorname{Top}_{q}(R).
\end{equation}
Here, $\operatorname{Top}_{q}$ returns the locations of the largest $q$ fraction of entries within each layer.

\parab{Step 3: fine-tuning.}
This step undoes the effect of the per-column scaling in $D_0$ by estimating its diagonal entries $s_i$. We compare the length between column vector $\hat{w}^i$ and $w^{\sigma(i)}_{\mathrm{pre}}$ to recover an approximation of $s_i$. Let $l(\cdot)$ be a length function of vector. The approximation of $s_i$ is computed as $\hat{s}_i = l(w^{\sigma(i)}_{\mathrm{pre}})/l(\hat{w}^i)$. 
We then initialize the approximate vector as $w^{\sigma(i)}_{\mathrm{init}} = \hat{s}_i \cdot \hat{w}^i$, and construct the initialization weight matrix $W_{\mathrm{init}} = [w^{\sigma(1)}_{\mathrm{init}},\allowbreak w^{\sigma(2)}_{\mathrm{init}},\allowbreak \ldots,\allowbreak w^{\sigma(n)}_{\mathrm{init}}]$. We then construct $W_{\mathrm{init}}$ for each layer to get the initialized model $\mathcal{M}_\textsf{init}$. The entries in the candidate support $\hat\Omega$ are initialized from $W_{\mathrm{pre}}$, after which we fine-tune $\mathcal{M}_\textsf{init}$ on $\mathcal{D}_{\mathrm{task}}$ to obtain $\mathcal{M}_\textsf{sur}$.
\section{Extending the Security Boundary}
\label{sec:defense}\label{sec:our-defense}

The leakage exposed by Collapse traces to a single invariant of
\priorboundary: \emph{the column directions of $W$
survive the obfuscation}. None of the primitives in $\mathcal{B}_{\mathrm{prior}}$ removes
this invariant. To push the boundary forward, the base set itself
must be enlarged with primitives that obfuscate column (and,
symmetrically, row) directions directly. \S~\ref{sec:defense:mix}
introduces the sparse multiplicative primitive
$\mathcal{P}_{M}$ that mixes columns of $W$, and
\S~\ref{sec:defense:ds} introduces the double-sided multiplicative
primitive $\mathcal{P}_{DS}$ that simultaneously mixes rows.
\S~\ref{sec:defense:boundary} writes down the canonical form
\extboundary of the extended family, and proves it is the new security boundary.

\subsection{Sparse Multiplicative Primitive}
\label{sec:defense:mix}

We break column atomicity by making each exposed column a sparse
combination of several victim columns, instead of a transformed copy
of one. 
Specifically, We seek a sparse, invertible multiplicative matrix whose inverse is
also sparse. A natural family is the sum of a constant number $k$
of permutation-diagonal terms,
\[
  M \;=\; \sum_{i=1}^{k} P_i D_i,
\]
where each $P_i$ is a permutation and each $D_i$ a non-singular
diagonal. Sampled appropriately, both $M$ and $M^{-1}$ admit the
same form, so each has $O(kn) = O(n)$ nonzero entries. The
$k=2$ case is sufficient to break column atomicity, and we use it
in our experiments. A concrete instantiation is given in
Appendix~\ref{app:111}. We refer to any matrix of this form as a \emph{generalized mixing matrix}.
The obfuscation primitive is defined as follows 
\[
  \mathcal{O}_W(W) \;=\; W M,
  \qquad
  \mathcal{R}(X \widetilde{W}) \;=\; (X \widetilde{W})\, M^{-1}.
\]
Recovery applies $M^{-1}$ as two paired permute-and-scale passes
over the columns of $X\widetilde{W}$ followed by a sum, at total cost
$O(mn)$---the same as $\mathcal{P}_D$.

\parab{Remark.}
The prior multiplicative primitives are degenerate cases of this
construction: setting $P_1 D_1 = 0$ gives $\mathcal{P}_D$ (with $P_0 = I$) and $\mathcal{P}_\Pi$ (with $D_0 = I$). So
$\mathcal{P}_{M}$ strictly subsumes
$\{\mathcal{P}_D, \mathcal{P}_\Pi\}$.

\subsection{Double-Sided Multiplicative Primitive}
\label{sec:defense:ds}

$\mathcal{P}_{M}$ mixes columns but leaves the row basis of
$X\widetilde{W}$ identical to that of $XW$; we close this attack channel by
mixing rows as well. Specifically, let $M_X, M_W$ be two independently sampled generalized mixing
matrices. The primitive is define as 
\begin{equation*}
\begin{aligned}
  \mathcal{O}_X(X) \;&=\; X\,M_X, \\
  \mathcal{O}_W(W) \;&=\; M_X^{-1}\,W\,M_W, \\
  \mathcal{R}(\widetilde{X}\widetilde{W}) \;&=\; (\widetilde{X}\widetilde{W})\,M_W^{-1}.
\end{aligned}
\end{equation*}
Correctness follows from
$\widetilde{X}\widetilde{W} = (X M_X)(M_X^{-1} W M_W) = X W M_W$,
which $\mathcal{R}$ recovers by stripping the trailing $M_W$; each
of the four sparse multiplications costs $O(mn)$.

\parab{Extending Definition~\ref{def:primitive}.}
The primitives in $\mathcal{B}_{\mathrm{prior}}$ and $\mathcal{P}_{M}$
all act on $W$ alone, fitting the pair $(\mathcal{O}_W, \mathcal{R})$
of Definition~\ref{def:primitive}. To accommodate a primitive that
also transforms $X$, we widen the definition to a triple
$(\mathcal{O}_X, \mathcal{O}_W, \mathcal{R})$ with the correctness
requirement $\mathcal{R}(\widetilde{X}\widetilde{W}) = XW$. Every
primitive in $\mathcal{B}_{\mathrm{prior}} \cup \{\mathcal{P}_{M}\}$
is the special case where $\mathcal{O}_X = \mathrm{id}$, so the extension
is strictly backwards-compatible, and the closure-under-composition
property of Theorem~\ref{thm:composition} carries over by stacking
$\mathcal{O}_X$ maps on the input side and $\mathcal{O}_W$ maps on
the weight side.

\parab{Remark.}
$\mathcal{P}_{DS}$ composes with any right-acting multiplicative
primitive: stacking $\mathcal{P}_{DS}$ around an existing
$\mathcal{P}_M, \mathcal{P}_D$, or $\mathcal{P}_\Pi$ yields a
double-sided variant of that primitive without any further
modification. In other words, every column-wise multiplicative
primitive admits a two-sided lift through $\mathcal{P}_{DS}$.

\subsection{The Extended Boundary \texorpdfstring{\extboundary}{O\_ext}}
\label{sec:defense:boundary}

Let $\mathcal{B}_{\mathrm{ext}}
  \;:=\;
  \mathcal{B}_{\mathrm{prior}} \,\cup\, \{\mathcal{P}_{M}, \mathcal{P}_{DS}\}$
and let $\mathcal{F}(\mathcal{B}_{\mathrm{ext}})$ be the family it
generates through (extended) composition. We claim the boundary of
this family is
\begin{equation}
\label{eq:oext}
\begin{aligned}
  \mathcal{O}_X(X) \;&=\; X\,M_X, \\
  \mathcal{O}_W(W) \;&=\; M_X^{-1}\,(W + L + S)\,M_W,
\end{aligned}
\end{equation}
where $S$ ranges over sparse matrices, $L$ over rank-$r$ matrices,
and $M_X, M_W$ over generalized mixing matrices. Setting $M_X = I$
and $M_W = D\Pi$ recovers \priorboundary, so \extboundary strictly subsumes \priorboundary.

\begin{theorem}[Security boundary of $\mathcal{B}_{\mathrm{ext}}$]
\label{thm:boundary-ext}
Every scheme\newline $\mathcal{P} \in \mathcal{F}(\mathcal{B}_{\mathrm{ext}})$
can be realized as
$\mathcal{O}_\textsf{ext}(\,\cdot\,;\,S, L, M_X, M_W)$ for some
choice of parameters.
 Here the additive masks satisfy
$\operatorname{nnz}(S) \leq s_{\max} \ll n^2$ and
$\operatorname{rank}(L) \leq r_{\max} \ll n$.
\end{theorem}

\begin{proof}
We extend the rewriting system of Theorem~\ref{thm:omax} to the
enlarged alphabet $\{S, L, D, \Pi, M_X, M_\mathrm{DS}\}$. Each
$\mathrm{DS}$ symbol expands into a matched pair: an $M_X$ acting
on the right of $X$, and an $M_X^{-1}$ acting on the left of $W$.
Two new rewriting rules suffice:
\begin{align}
M^{(1)} M^{(2)} &\;\longrightarrow\; M^{(12)} \label{eq:rw-mm}\\
M_X^{-1}\,(W + L + S) &\;\longrightarrow\; M_X^{-1} W + L'' + S''
  \label{eq:rw-leftdist}
\end{align}
Rule~\eqref{eq:rw-mm} reflects that any product of generalized
mixing matrices remains a mixing matrix: two $k$-term mixing
matrices multiply to a $k^2$-term one, but since $\mathcal{P}$ has
bounded length, the accumulated term count stays $O(1)$ in $n$, so
the result is still $O(n)$-sparse. Rule~\eqref{eq:rw-leftdist}
is the left-side analogue of
Rules~\eqref{eq:rw-lm}--\eqref{eq:rw-sm}: $M_X^{-1} L$ has rank
at most $r$, and $M_X^{-1} S$ has $O(\mathrm{nnz}(S))$ nonzeros.
Note also that $D$ and $\Pi$ are degenerate mixing matrices
(\S~\ref{sec:defense:mix}), so the old
Rules~\eqref{eq:rw-dd}--\eqref{eq:rw-pid} are absorbed by
Rule~\eqref{eq:rw-mm}.

The remaining three phases mirror Theorem~\ref{thm:omax}:
distribute multiplicative factors through additive structure on
both sides, fuse the right-side chain into one $M_W$ via
Rule~\eqref{eq:rw-mm} and the left-side factors into one
$M_X^{-1}$ analogously, then fuse same-type additive masks via
Rules~\eqref{eq:rw-ll} and~\eqref{eq:rw-ss}. Absorbing the
resulting $L^*, S^*$ into the parenthesized block by the
substitutions $\widehat{L} = M_X L^* M_W^{-1}$ and
$\widehat{S} = M_X S^* M_W^{-1}$ (both still rank-$r$ and sparse)
yields the canonical form
$M_X^{-1}(W + \widehat{S} + \widehat{L}) M_W$ of
Eq.~\eqref{eq:oext}.
 Thus, every efficiency-preserving scheme admitted to
$\mathcal{F}(\mathcal{B}_{\mathrm{ext}})$ satisfies these sparsity and
rank thresholds.
\end{proof}

\parab{Cost.}
The per-query overhead of \extboundary stays within
the \priorboundary envelope: applying $M_X$ and
$M_W^{-1}$ are $O(mn)$ sparse-dense products, identical to
$\mathcal{P}_D$. The per-refresh cost of materializing $M_X^{-1}(W + L + S) M_W$ inside 
the TEE is $O(n^2)$, on par with \priorboundary.

\section{Evaluation}
\label{sec:eval}

In this section we evaluate the performance of \sysattack and \extboundary. We first describe the evaluation setup (\S~\ref{sec:eval_setup}) and then answer the following research questions.

\begin{itemize}[leftmargin=2em]
\setlength\itemsep{0.15em}
\item \textbf{RQ1:} How effective is \sysattack against existing defenses?
\item \textbf{RQ2:} How accurately does \sysattack recover the target obfuscation components at each stage?
\item \textbf{RQ3:} How effective is \extboundary as a defense against \sysattack?
\item \textbf{RQ4:} What cost does \extboundary impose relative to prior defenses?
\end{itemize}

\subsection{Experimental Setup}
\label{sec:eval_setup}

\noindent\textbf{Models.}
As summarized in Table~\ref{tab:model_specs}, we evaluate \sysattack and \extboundary across four representative models to ensure broad applicability. We select BERT-Base~\cite{devlin2019bert} and the Qwen2.5 series~\cite{qwen2024qwen25} to cover both auto-encoding and state-of-the-art auto-regressive (decoder-only) architectures. ViT-Base~\cite{dosovitskiy2021vit} is further included to verify generalizability on vision-related transformer tasks. To assess scalability, our selection spans from $86$M parameters up to $1.54$B (Qwen2.5-1.5B), the latter mirroring the complexity of modern high-performance edge LLMs such as Microsoft's Phi~\cite{abdin2024phi3}. All model checkpoints are loaded via the HuggingFace Transformers library~\cite{wolf2020transformers}.
\begin{table}[htbp]
\centering
\caption{Summary of model architectures and configurations used in our evaluation.}
\label{tab:model_specs}
\resizebox{\linewidth}{!}{
\begin{tabular}{lcccc}
\toprule
\textbf{Model} & \textbf{Architecture} & \textbf{Params} & \textbf{Hidden ($n$)} & \textbf{Layers ($L$)} \\ \midrule
BERT-Base      & Auto-encoding         & 110M            & 768                   & 12                    \\
ViT-Base       & Vision Transformer    & 86M             & 768                   & 12                    \\
Qwen2.5-0.5B   & Auto-regressive       & 490M            & 896                   & 24                    \\
Qwen2.5-1.5B   & Auto-regressive       & 1.54B           & 1536                  & 28                    \\ \bottomrule
\end{tabular}
}
\end{table}

\noindent\textbf{Datasets.}
We evaluate four representative datasets that cover multiple modalities and task complexities, summarized in Table~\ref{tab:dataset_specs}. For natural language understanding, we select MNLI~\cite{williams2018mnli}, SST-2~\cite{socher2013sst}, and QNLI~\cite{rajpurkar2016squad} from the General Language Understanding Evaluation (GLUE) benchmark~\cite{wang2018glue}. MNLI is a large-scale resource for multi-genre natural-language inference, SST-2 is the industry standard for sentiment analysis, and QNLI converts question answering into a sentence-pair classification task. To demonstrate generalizability in computer vision, we evaluate ViT-Base on CIFAR-100~\cite{krizhevsky2009cifar}. This selection ensures that our evaluation reflects realistic deployment scenarios for models shielded by trusted execution environments.

\begin{table}[htbp]
\centering
\caption{Representative datasets used in our evaluation.}
\label{tab:dataset_specs}
\resizebox{\linewidth}{!}{
\begin{tabular}{lcccl}
\toprule
\textbf{Dataset} & \textbf{Modality} & \textbf{Task Type} & \textbf{\# Samples} & \textbf{Evaluation Metric} \\ \midrule
MNLI             & Text              & NLI                & 393K                & Accuracy                   \\
SST-2            & Text              & Sentiment          & 67K                 & Accuracy                   \\
QNLI             & Text              & QA/NLI             & 105K                & Accuracy                   \\
CIFAR-100        & Image             & Classification     & 60K                 & Top-1 Accuracy             \\ \bottomrule
\end{tabular}
}
\end{table}

\noindent\textbf{Selective Defenses.}
To evaluate the effectiveness of \sysattack against various protection paradigms, we select five representative obfuscation algorithms summarized in Table~\ref{tab:related_works}. For each prior defense, we evaluate the matrix-level weight transformation summarized in Table~\ref{tab:related_works}. We include NNSplitter and LoRO to represent the additive category, which use sparse masks and low-rank masks respectively. TSQP and TransLinkGuard cover multiplicative primitives via column-wise scaling and permutation. To assess composite defenses, we incorporate ArrowCloak~\cite{wang2025arrows}, the state-of-the-art fusion method combining both multiplicative and additive transformations. We additionally evaluate \sysattack against \priorboundary, the canonical form of compositions over the $\{S,L,D,\Pi\}$ primitives defined in \S~\ref{sec:security_boundary}.

\noindent \textbf{Key-refresh configuration.}
Following the refresh settings of the prior schemes summarized in
Table~\ref{tab:related_works}, we keep the sparse mask $S$ static and refresh
$L_t$, $D_t$, and $\Pi_t$ at each key epoch.

\noindent \textbf{Baselines.}
We compare \sysattack against three baselines to provide a clear performance reference. The \emph{White-box} setting assumes the defender deploys the model directly to an untrusted REE. The adversary possesses complete access to the model parameters and architecture, and obtains $\mathcal{M}_\textsf{vic}$ directly without any recovery effort. We fine-tune $\mathcal{M}_\textsf{vic}$ with the same $1\%$-task-data budget that \sysattack uses. This establishes the upper bound of attack performance. The \emph{Black-box} setting represents the scenario where model execution is fully protected within the TEE. The adversary only has access to the public pre-trained model $\mathcal{M}_\textsf{pre}$ and a limited training set ($1\%$-task-data, the same budget as \sysattack). We fine-tune all parameters of $\mathcal{M}_\textsf{pre}$ to produce a surrogate $\mathcal{M}_\textsf{sur}$. This baseline demonstrates that generic fine-tuning alone fails to reconstruct $\mathcal{M}_\textsf{sur}$ to the white-box level, underscoring the necessity of the model stealing attacks. 
ArrowMatch~\cite{wang2025arrows} serves as the state-of-the-art prior attack against multiplicative obfuscation primitives ($\{\Pi,D\}$). Compared to ArrowMatch, \sysattack exhibits a significant advantage in handling additive primitives ($\{S,L\}$). This performance gap validates the advanced capabilities of our attack pipeline in recovering parameters from the secured models. For \extboundary we additionally compare against the strongest prior defense, ArrowCloak~\cite{wang2025arrows}, to quantify the performance reduction in \sysattack surrogate accuracy from the new primitives.

\noindent\textbf{Metric.}
We evaluate attack effectiveness using two metrics. First, we measure the classification accuracy of the surrogate model $\mathcal{M}_\textsf{sur}$. A higher accuracy indicates that $\mathcal{M}_\textsf{sur}$ has captured a greater portion of the functionality from $\mathcal{M}_\textsf{vic}$. 
Second, we calculate
\[
\text{Rel.Black}=\frac{\mathrm{Acc}(\mathcal{M}_\textsf{sur})}{\mathrm{Acc}(\text{Black-box})},
\]
the relative performance compared to the black-box baseline. This comparison highlights how effectively our attack pipeline accelerates the model fine-tuning phase. We report both metrics across all model--task--defense combinations.

\subsection{RQ1: Attack Performance Across Defenses}
\label{sec:eval_rq1}

  \begin{table*}[!t]                                                                                                                                 
  \centering                                             
  \caption{Attack accuracy (\%) across four models, four tasks, and six evaluated weight-obfuscation configurations; single-seed. $\mathcal{M}_\textsf{obf}$: exposed obfuscated model;
  \sysattack: recovered surrogate. White-box / Black-box are upper/lower bounds.}                            
  \label{tab:new_attack_results}
  \setlength{\tabcolsep}{2pt}                                                                                                                        
  \resizebox{\textwidth}{!}{                                                                                                                         
  \begin{tabular}{llcccccccccccccc}
  \toprule                                                                                                                                           
   & & \multicolumn{2}{c}{NNSplitter} & \multicolumn{2}{c}{LoRO} & \multicolumn{2}{c}{TSQP} & \multicolumn{2}{c}{TransLinkGuard} &
  \multicolumn{2}{c}{ArrowCloak} & \multicolumn{2}{c}{\priorboundary} & \multirow{2}{*}{White-box} & \multirow{2}{*}{Black-box} \\                   
  \cmidrule(lr){3-4} \cmidrule(lr){5-6} \cmidrule(lr){7-8} \cmidrule(lr){9-10} \cmidrule(lr){11-12} \cmidrule(lr){13-14}
   & & $\mathcal{M}_\textsf{obf}$ & \sysattack & $\mathcal{M}_\textsf{obf}$ & \sysattack & $\mathcal{M}_\textsf{obf}$ & \sysattack & $\mathcal{M}_\textsf{obf}$ & \sysattack &               
  $\mathcal{M}_\textsf{obf}$ & \sysattack & $\mathcal{M}_\textsf{obf}$ & \sysattack & & \\                                                                               
  \midrule                                                                                                                                           
  ViT-Base & CIFAR-100 & 1.94 & 90.28 & 1.14 & 90.36 & 0.88 & 89.44 & 0.98 & 90.26 & 0.88 & 89.82 & 0.86 & 90.62 & 90.44 & 68.44 \\                  
  \midrule                                                                                                                                           
  \multirow{3}{*}{BERT-Base} & MNLI & 33.01 & 84.67 & 32.13 & 84.86 & 34.77 & 84.47 & 32.03 & 84.57 & 32.03 & 83.59 & 32.03 & 83.69 & 84.67 & 45.80  
  \\                                                                                                                                                 
   & QNLI & 49.80 & 90.43 & 50.20 & 90.62 & 48.83 & 90.62 & 53.91 & 91.02 & 53.91 & 90.62 & 49.31 & 89.97 & 90.43 & 63.48 \\
   & SST-2 & 49.02 & 93.16 & 53.52 & 93.16 & 51.56 & 92.77 & 53.52 & 93.16 & 53.52 & 92.58 & 53.52 & 93.36 & 93.16 & 58.20 \\                        
  \midrule                                                                                                                                           
  \multirow{3}{*}{Qwen2.5-0.5B} & MNLI & 34.77 & 83.79 & 33.98 & 84.77 & 33.59 & 84.96 & 32.42 & 84.18 & 32.03 & 83.40 & 31.84 & 83.40 & 84.38 &     
  66.41 \\                                                                                                                                           
   & QNLI & 50.00 & 89.84 & 52.15 & 89.26 & 46.29 & 91.41 & 51.37 & 89.65 & 47.46 & 91.60 & 49.22 & 89.84 & 91.99 & 62.11 \\
   & SST-2 & 49.02 & 94.34 & 48.24 & 93.75 & 49.61 & 91.60 & 48.24 & 94.14 & 50.39 & 94.14 & 47.07 & 92.38 & 94.14 & 51.17 \\                        
  \midrule                                                                                                                                           
  \multirow{3}{*}{Qwen2.5-1.5B} & MNLI & 34.18 & 89.26 & 33.79 & 89.26 & 32.42 & 90.23 & 31.05 & 89.65 & 32.03 & 88.48 & 33.25 & 89.28 & 91.41 &     
  72.36 \\                                                                                                                                           
   & QNLI & 47.07 & 93.16 & 53.32 & 93.36 & 54.10 & 94.14 & 53.52 & 92.97 & 52.93 & 92.97 & 48.37 & 92.26 & 93.75 & 67.58 \\
   & SST-2 & 47.27 & 96.48 & 52.15 & 96.88 & 52.54 & 95.70 & 53.32 & 96.29 & 54.49 & 96.09 & 51.43 & 94.86 & 96.68 & 69.53 \\                        
  \midrule                                                                                                                                           
    Average & & 39.61 & 90.54 & 41.07 & 90.63 & 40.46 & 90.53 & 41.04 & 90.59 & 40.97 & 90.33 & 39.69 & 89.97 & 91.11 & 62.51 \\ 
  Ratio   & & 0.63$\times$ & 1.45$\times$ & 0.66$\times$ & 1.45$\times$ & 0.65$\times$ & 1.45$\times$ & 0.66$\times$ & 1.45$\times$ & 0.66$\times$ & 1.45$\times$ & 0.63$\times$ & 1.44$\times$ & 1.46$\times$ & 1.00$\times$ \\   
  \bottomrule                                                                                                                                  
  \end{tabular}                                          
  }
  \end{table*}

We report the attack performance of \sysattack in Table~\ref{tab:new_attack_results}. The last row reports the average performance across all settings. We can observe that \sysattack effectively recovers the obfuscated weights and achieves a high attack performance. Relative to the black-box baseline, the attack reaches on average $1.45\times$, essentially matching the $1.46\times$ white-box upper bound. On \priorboundary, the full-stack target, the surrogate is also close to white-box. This shows that \sysattack successfully recovers the knowledge encoded in the obfuscated weights, lifting the attack performance from the black-box level to near-white-box.

On the contrary, the $\mathcal{M}_\textsf{obf}$-based baseline has a much lower attack performance. We can also observe that $\mathcal{M}_\textsf{obf}$ performs even worse than a black box fine-tuning model. It is because the obfuscated weights $\mathbf{W}_{\mathrm{obf}}$ give a destructive initialization for $\mathcal{M}_\textsf{sur}$, this initialization is even worse than fine-tuning from $\mathbf{W}_{\mathrm{pre}}$ alone. Thus $\mathbf{W}_{\mathrm{obf}}$ cannot be used to train a surrogate model. It means that the evaluated weight-obfuscation configurations are effective in front of the naive $\mathcal{M}_\textsf{obf}$-based attack. But \sysattack can utilize the leaked column information to recover weights from these configurations and achieve near-white-box performance.

\noindent\emph{Answer to RQ1.} \sysattack effectively recovers weights from the evaluated configurations and achieves a high attack performance. The attack performance is on average $1.45\times$ the black-box baseline (the white-box upper bound is $1.46\times$).

\subsection{RQ2: Stage-wise Attack Success}
\label{sec:eval_rq2}

We validate the complete recovery pipeline by comparing the output of each
stage with the corresponding ground truth. The results show that
\sysattack succeeds at every stage rather than relying solely on downstream
fine-tuning. We conduct three independent attack runs using BERT-Base on MNLI and report the mean of each metric in Table~\ref{tab:stage_recovery}.

Table~\ref{tab:stage_recovery} reports the recovery metrics for each stage.
In Stage~1, we report the percentage of columns correctly aligned across
views. In Stage~2, we report $L$-Sim ($1-\operatorname{RelF}_L$), which
measures the similarity between the recovered and ground-truth $L$. In
Stage~3, we report absolute-permutation accuracy and the precision, recall,
and F1 score for sparse-support estimation.

\begin{table}[htbp]

\centering
\caption{Stage-wise recovery performance of \sysattack against
\priorboundary. Results are the mean of three independent attack runs.}
\label{tab:stage_recovery}
\setlength{\tabcolsep}{2.5pt}
\resizebox{\linewidth}{!}{
\begin{tabular}{lcccc}
\toprule
\textbf{Setting}
& \makecell{\textbf{Stage 1}\\Rel. $\Pi$ Acc.}
& \makecell{\textbf{Stage 2}\\$L$-Sim.}
& \makecell{\textbf{Stage 3}\\Abs. $\Pi$ Acc.}
& \makecell{\textbf{Stage 3}\\Support P/R/F1} \\
\midrule
BERT-Base/MNLI & 100.00\% & 99.97\% & 100.00\% & 99.36\%/99.22\%/99.29\% \\
\bottomrule
\end{tabular}}

\end{table}

\noindent\emph{Answer to RQ2.} \sysattack achieves approximately 100\%
recovery performance at every stage.

\subsection{RQ3: Defense Effectiveness}
\label{sec:eval_rq3}

In this RQ, we compare \extboundary's defense effectiveness with prior obfuscation methods. We evaluate \extboundary using the existing \sysattack pipeline and the same hyperparameters as in \S~\ref{sec:eval_rq1}, without a defense-specific adaptation to \extboundary. Table~\ref{tab:defense_performance} shows the attack performance against \extboundary and the best performance of the existing lightweight obfuscation algorithms (denoted as ``Prior Best''). We report $\mathcal{M}_\textsf{sur}$'s accuracy and the relative accuracy compared to the black-box baseline (denoted as ``Rel.Black''). The black-box baseline is an empirical reference point obtained without using the obfuscated weights, rather than a formal security lower bound or ideal-protection guarantee.

Under the evaluated \sysattack configuration, \extboundary yields lower surrogate accuracy than the best prior defense. On average, the Rel.Black of \extboundary is $1.00\times$, but Prior Best is $1.46\times$. The surrogate reaches the empirical black-box reference point on average, indicating that this evaluated pipeline obtains no measured benefit in average accuracy from the obfuscated weights. Notably, \extboundary performs even better on smaller models such as ViT-Base, where Rel.Black drops to $0.78\times$. These results establish empirical model-stealing resistance only under the evaluated \sysattack configuration.

\begin{table}[htbp]
\centering
\caption{Surrogate accuracy under \extboundary vs.\ the strongest prior defense, both evaluated against \sysattack. Rel.Black is $\mathrm{Acc}(\mathcal{M}_\textsf{sur})/\mathrm{Acc}(\text{Black-box})$.}
\label{tab:defense_performance}
\resizebox{\linewidth}{!}{
\begin{tabular}{llcccc}
\toprule
\multirow{2.5}{*}{Model} & \multirow{2.5}{*}{Dataset} & \multicolumn{2}{c}{\extboundary} & \multicolumn{2}{c}{Prior Best} \\
\cmidrule(lr){3-4} \cmidrule(lr){5-6}
& & $\mathcal{M}_\textsf{sur}$ & Rel.Black & $\mathcal{M}_\textsf{sur}$ & Rel.Black \\
\midrule
ViT-Base & CIFAR-100 & 53.40\% & 0.78$\times$ & 89.44\% & 1.31$\times$ \\
\midrule
\multirow{2}{*}{BERT-Base} & MNLI & 35.86\% & 0.78$\times$ & 83.59\% & 1.83$\times$ \\
& SST-2 & 59.97\% & 1.03$\times$ & 92.58\% & 1.59$\times$ \\
\midrule
\multirow{3}{*}{Qwen2.5-0.5B} & MNLI & 66.99\% & 1.01$\times$ & 83.40\% & 1.26$\times$ \\
& QNLI & 70.75\% & 1.14$\times$ & 89.26\% & 1.44$\times$ \\
& SST-2 & 66.94\% & 1.31$\times$ & 91.60\% & 1.79$\times$ \\
\midrule
\multirow{3}{*}{Qwen2.5-1.5B} & MNLI & 74.74\% & 1.03$\times$ & 88.48\% & 1.22$\times$ \\
& QNLI & 67.48\% & 1.00$\times$ & 92.97\% & 1.38$\times$ \\
& SST-2 & 63.96\% & 0.92$\times$ & 95.70\% & 1.38$\times$ \\
\midrule
Average & & 62.23\% & 1.00$\times$ & 89.67\% & 1.46$\times$ \\
\bottomrule
\end{tabular}
}
\end{table}

\noindent\emph{Answer to RQ3.} Under the evaluated \sysattack configuration, \extboundary reduces Rel.Black from $1.46\times$ under the strongest prior defense (ArrowCloak) to $1.00\times$, matching the empirical black-box reference on average. This result is empirical and attack-specific.

\subsection{RQ4: Cost of \texorpdfstring{\extboundary}{O\_ext}}
\label{sec:eval_rq4}\label{sec:eval:cost}

\begin{table}[htbp]
\centering
\caption{Phase-level latency (ms) for a protected BERT-Base block with $b=1$.}
\label{tab:phase_overhead}
\resizebox{\linewidth}{!}{
\begin{tabular}{l cccc cccc c}
\toprule
\multirow{2}{*}{Defense} & \multicolumn{4}{c}{Self-Attention (ms)} & \multicolumn{4}{c}{Feed-Forward (ms)} & \multirow{2}{*}{Tot.\ (ms)} \\
\cmidrule(lr){2-5} \cmidrule(lr){6-9}
 & Obf.\ (ms) & GPU (ms) & Trans.\ (ms) & Rec.\ (ms) & Obf.\ (ms) & GPU (ms) & Trans.\ (ms) & Rec.\ (ms) & \\
\midrule
ArrowCloak           & 24.15 & 5.54 & 250.99 & 313.92 & 42.36 & 0.60 & 138.63 & 146.63 & 922.81 \\
\priorboundary       & 22.12 & 4.66 & 202.99 & 318.46 & 40.33 & 0.49 & 118.09 & 148.70 & 855.82 \\
\extboundary         & 80.28 & 4.38 & 195.94 & 303.80 & 59.16 & 0.48 & 116.08 & 143.21 & 903.33 \\
\bottomrule
\end{tabular}
}
\end{table}

\begin{table}[htbp]
\centering
\caption{Per-primitive latency on a $768\times 768$ BERT-Base linear layer at $b=1$, seq.\ length $768$.}
\label{tab:operator_overhead}
\resizebox{\linewidth}{!}{
\begin{tabular}{lccc}
\toprule
Primitive & Obfuscation (ms) & GPU Compute (ms) & Recovery (ms) \\
\midrule
Sparse Mask & 1.010 & 0.242 & 7.428 \\
Low-rank Mask ($r=1/2/4$) & 1.04/1.04/1.00 & 0.25/0.26/0.24 & 3.05/5.93/9.51 \\
Col.\ Scaling / Permutation & 0.82 / 3.96 & 0.23 / 0.23 & 0.87 / 4.12 \\
\midrule
Sparse Mixing & 10.97 & 0.23 & 11.61 \\
\bottomrule
\end{tabular}
}
\end{table}

We implement a TEE-GPU prototype to measure the obfuscation and recovery cost on the TEE side. The TEE runs Python inside Occlum, the GPU runs a long-lived worker process, and the two sides communicate via sockets. We run two microbenchmarks on a $768\times 768$ BERT-Base linear layer at $b=1$ and sequence length $768$ ($10$-sample mean), and break each measurement into four phases: obfuscation, GPU compute, TEE$\leftrightarrow$GPU transfer, and TEE-side recovery. Table~\ref{tab:operator_overhead} reports the cost of each primitive in isolation; the Double-Sided primitive is omitted because its cost is by construction twice that of its column-wise multiplicative counterpart. Table~\ref{tab:phase_overhead} aggregates the four phases over one self-attention block and one feed-forward block for \extboundary, ArrowCloak~\cite{wang2025arrows} (the strongest prior hybrid scheme), and \priorboundary (the canonical representative of the $\{S,L,D,\Pi\}$ family).

\noindent\textbf{New primitives are on par with prior ones.} Table~\ref{tab:operator_overhead} shows that the sparse multiplicative primitive stays within the same order of magnitude as the prior primitives. This is because it can be decomposed into a permutation $\Pi$ followed by a diagonal scaling $D$, so its cost is essentially that of $\Pi$ plus $D$.

\noindent\textbf{\extboundary's end-to-end cost matches prior schemes.} Table~\ref{tab:phase_overhead} shows that \extboundary's end-to-end cost is on par with both \priorboundary and ArrowCloak. This is because the new primitives only add cost in the obfuscation phase for $\mathbf{X}$, while recovery and transfer are essentially unaffected.

\noindent\emph{Answer to RQ4.} \extboundary's end-to-end cost is on par with ArrowCloak while reducing Collapse surrogate accuracy by $27.4$ percentage points (\S~\ref{sec:eval_rq3}).

\section{Discussion}

\noindent\textbf{GPU TEEs.}
GPU-based confidential computing is a promising direction for protecting
accelerator-side execution, with proposals such as Graviton~\cite{volos2018graviton}
and ACAI~\cite{sridhara2024acai} extending hardware-isolated execution to the
accelerator (see~\cite{wang2024ccgpu} for a survey of the CPU--GPU
confidential-computing landscape).
However, it depends on specific hardware, firmware,
and deployment support, and is not yet available on many commodity devices. Our
work studies the current TEE--REE setting where the CPU-side TEE is trusted
while the GPU/NPU remains untrusted. Thus, our primitive-based analysis is
complementary to GPU TEEs: confidential GPUs can reduce the exposed boundary when available, while our results characterize the security of algorithmic obfuscation when the external accelerators are not trusted.

\noindent\textbf{Side-channel attacks.}
We assume the TEE correctly protects en\-clave-resident code and data, and do not
consider hardware side channels or transient-execution attacks against the TEE.
A separate body of work recovers model weights through fault injection or
memory/PCIe side channels---e.g., RowHammer-based weight extraction in
DeepSteal~\cite{rakin2022deepsteal}, GPU memory-bus snooping in
HyperTheft~\cite{yuan2024hypertheft} and CipherSteal~\cite{yuan2025ciphersteal},
PCIe-traffic analysis in the Hermes attack~\cite{zhu2021hermes}, and the
side-channel reverse-engineering of CNNs by Hua et al.~\cite{hua2018reverse}.
This line is orthogonal to our focus: our attacks do not rely on compromising
the TEE or any hardware channel; they exploit algebraic structure in the
obfuscated tensors already exposed to the REE. Existing TEE and bus-side
side-channel mitigations can therefore be combined with our framework.

\noindent\textbf{Key rotation and boundary-visible leakage.}
Our analysis reveals a counterintuitive consequence of lightweight
obfuscation. Because these primitives conceal only selected algebraic aspects
of the weight matrix, a single obfuscated view can still reveal exploitable
structure. Refreshing the keys produces multiple differently obfuscated views;
comparing them may reveal complementary clues and make weight recovery easier.
Thus, under our threat model, more frequent key rotation can weaken resistance
to model extraction. Appendix~\ref{app:fully-refreshed-attack} analyzes the
fully refreshed setting.

\noindent\textbf{Defense scope and adaptive attacks.}
Our evidence for \extboundary is empirical and limited to the existing \sysattack pipeline without a defense-specific adaptation.
Developing defense-aware attacks targeting $M_X$ and $M_W$ is left for future work. Thus, we do not claim any form of formal model-extraction security guarantee with \extboundary.

\noindent\textbf{Other TSLP Schemes.}
GroupCover~\cite{zhang2024groupcover} represents a different security--efficiency trade-off. By mixing
groups of weight vectors through matrix multiplication, it can hide stronger
column-level invariants than efficient TSLP schemes.
However, it splits into mixing groups based on distances between weight vectors,
rather than randomly choosing groups. Since fine-tuned weights are often close
to the public pretrained weights, $W_{\mathrm{pre}}$ may help the attacker guess which vectors are mixed together.
ShadowNet~\cite{sun2023shadownet} uses a dense additive mask, but the mask is not fully confined within the TEE: part of the mask-related information is exposed to the GPU to enable efficient outsourced computation.
This exposure gives the attacker extra information to steal the protected model.

\vspace{-0.3cm}
\section{Related Work}

\noindent\textbf{TEE-shielded model partition.}
TEE-shielded partition schemes split a model between TEE and REE; early designs such as DarkneTZ~\cite{mo2020darknetz} and Occlumency~\cite{lee2019occlumency} simply partition without obfuscating offloaded weights. To better protect the offloaded weights, subsequent schemes combine algebraic tricks heuristically: ShadowNet~\cite{sun2023shadownet} composes permutation with random per-channel scaling on convolution filters; Magnitude~\cite{hou2022model} partitions selected high-magnitude weights into the TEE; NNSplitter~\cite{zhou2023nnsplitter} perturbs a small fraction ($\approx 0.002\%$) of weight entries and stores their original values in the TEE; TransLinkGuard~\cite{li2024translinkguard} couples layer-specific row/column permutations with TEE authorization that carries permuted activations across adjacent transformer layers; Soter~\cite{shen2022soter} morphs operators with scalar blinding coins and restores accumulated coins at enclave boundaries; TSQP~\cite{sun2025tsqp} combines a learned dissimilar model with quantization scaling; LoRO~\cite{xiong2025loro} obfuscates weights with a dense mask generated from low-rank factors; ArrowCloak~\cite{wang2025arrows} composes permutation, scaling, and a low-rank mask. A common primitive-level representation 
of these schemes is missing.

\noindent\textbf{Attacks on obfuscation-based TSLP.}
Existing attacks rely heavily on expert intuition and lack a unified methodology. Zhang et al.~\cite{zhang2024noprivacy} taxonomize TSDP schemes by which layers are shielded, without analyzing the algebraic structure of the obfuscation applied to offloaded layers. ArrowMatch~\cite{wang2025arrows} attacks schemes built on permutation and scaling by matching obfuscated columns against a public prior, with an attack design specific to these two primitives. LoRO~\cite{xiong2025loro} similarly attacks permutation-, scaling-, and sparse-mask-based schemes by exploiting their statistical proximity to a public prior, but its attack design relies heavily on hand-crafted heuristics.

\noindent\textbf{Our work.}
We move the design of TSLP attacks and defenses beyond expert intuition. To this end, we introduce a unified algebraic abstraction over the representative obfuscation primitives in prior approaches. From this abstraction we derive a canonical structural form for the obfuscation family compossible by these primitives and a systematic attack against it, and extend the primitive set with two new constructs that close the structural channels exploited by the attack. 
We hope to provide a systematic way of understanding obfuscation-based secure LLM inference, so that our community can make reasonable claims about their approaches. 
 
\vspace{-0.3cm}
\section{Conclusion}
The rapid evolution of heuristic attacks and defenses for TEE-shielded LLM partitioning (TSLP) has substantially advanced this area, but it has also led to a fragmented design space dominated by ad-hoc solutions. As a result, existing schemes are difficult to compare, evaluate, and extend in a principled manner. In this paper, we take a step toward systematizing this design space by introducing \emph{unified obfuscation primitives}, a primitive-level abstraction that captures representative defense mechanisms for TSLP. This abstraction enables us to identify a canonical structural boundary for the declared primitive family. Building on this foundation, we systematically attack this canonical form to reveal the shared vulnerabilities of many existing TSLP approaches published in top-tier venues. 
Our evaluation demonstrates that the proposed attack can effectively recover the hidden obfuscation structure and build high-fidelity surrogate models across diverse obfuscation families, revealing the fundamental limitations of lightweight obfuscation in TSLP systems.

\begin{acks}
We thank our anonymous reviewers for their insightful feedback. The research is supported in part by the National Key R\&D Program of China under Grant 2024YFB2906803, and National Natural Science Foundation of China (NSFC) under Grant 62472247, as well as a CIE-Smartchip research grant. The corresponding author of this paper is Zhuotao Liu.
\end{acks}

\bibliographystyle{ACM-Reference-Format}
\bibliography{gab}

\appendix

\section{Open Science}
\label{app:open-science}
Our artifacts are available at
\url{https://github.com/InspiringGroup-NeoLab/TEE-Obfuscation}.

\section{Ethical Considerations}
\label{app:ethics}
This dual-use study informs stronger defenses; we evaluate \sysdefense
as a mitigation. The adversary controls the device without compromising
the TEE; cloud/API-only deployments are outside scope. Experiments use
public checkpoints, with no live targets or sensitive personal data.
Artifacts provide attack and defense implementations, not extracted
victim parameters.

\section{Sparse Mixing Matrix Construction}
\label{app:111}
The sparse mixing operator in \S~\ref{sec:our-defense} requires an invertible matrix $\mathbf{M}$ such that both $\mathbf{M}$ and $\mathbf{M}^{-1}$ remain sparse. A simple way to satisfy both requirements is to use constant-size block mixing.

Partition the column indices into groups $\mathcal{G}_1,\ldots,\mathcal{G}_s$ with $|\mathcal{G}_j|=k_j\leq k$, where $k$ is a small constant ($k=2$ in our setting). For each group, choose an invertible matrix $\mathbf{B}_j\in\mathbb{R}^{k_j\times k_j}$. After applying an input permutation $\boldsymbol{\Pi}_{\mathrm{in}}$ to gather the columns of the same group together, we construct
\[
\mathbf{M}
=
\boldsymbol{\Pi}_{\mathrm{out}}
\begin{bmatrix}
\mathbf{B}_1 & \mathbf{0} & \cdots & \mathbf{0} \\
\mathbf{0} & \mathbf{B}_2 & \cdots & \mathbf{0} \\
\vdots & \vdots & \ddots & \vdots \\
\mathbf{0} & \mathbf{0} & \cdots & \mathbf{B}_s
\end{bmatrix}
\boldsymbol{\Pi}_{\mathrm{in}},
\]
with $\boldsymbol{\Pi}_{\mathrm{in}},\boldsymbol{\Pi}_{\mathrm{out}}$ permutation matrices. Then
\[
\mathbf{M}^{-1}
=
\boldsymbol{\Pi}_{\mathrm{in}}^{-1}
\begin{bmatrix}
\mathbf{B}_1^{-1} & \mathbf{0} & \cdots & \mathbf{0} \\
\mathbf{0} & \mathbf{B}_2^{-1} & \cdots & \mathbf{0} \\
\vdots & \vdots & \ddots & \vdots \\
\mathbf{0} & \mathbf{0} & \cdots & \mathbf{B}_s^{-1}
\end{bmatrix}
\boldsymbol{\Pi}_{\mathrm{out}}^{-1}.
\]
Because each block has size at most $k=O(1)$, both $\mathbf{M}$ and $\mathbf{M}^{-1}$ contain only $O(kn)=O(n)$ nonzero entries. Obfuscation and recovery therefore require only linear-time sparse post-processing. One concrete instantiation uses $2\times 2$ blocks
$\mathbf{B}_j=\bigl[\begin{smallmatrix}1&\alpha_j\\\beta_j&1+\alpha_j\beta_j\end{smallmatrix}\bigr]$ with $\det\mathbf{B}_j=1$, making every exposed column a linear combination of two original columns.

\section{Model Extraction Game}
\label{app:extraction-game}

This appendix makes precise the model-extraction game summarized in
Definition~\ref{def:obfuscation-primitive}. We use the system and threat
model of \S~\ref{sec:threat_model} throughout.

\paragraph{Participants and deployment.}
The model owner fine-tunes the public pretrained model
$\mathcal{M}_{\textsf{pre}}$ into the proprietary victim model
$\mathcal{M}_{\textsf{vic}}$ and deploys it through the TEE--REE
pipeline using an obfuscation scheme $\mathcal{S}$. The scheme samples
and refreshes its keys according to its declared schedule. The
device-owner adversary is honest-but-curious: it follows the prescribed
inference protocol and cannot compromise the TEE, but it controls the
REE and records all information exposed there.

\paragraph{Available and hidden information.}
The adversary knows the model architecture,
$\mathcal{M}_{\textsf{pre}}$, and the public description and parameters
of $\mathcal{S}$. It may use labeled task data up to the declared budget
$B$. Across key refreshes, it records at most $T$ distinct exposed
views. For every offloaded layer $\ell$ in view $t$, this includes the
released input, the obfuscated weight, and the REE-computed product. We
write this layer-level observation as
\[
  \mathsf{V}_{\ell,t}
  :=
  \left(
    \mathbf{X}_{\ell,t},
    \widetilde{\mathbf{W}}_{\ell,t},
    \mathbf{X}_{\ell,t}\widetilde{\mathbf{W}}_{\ell,t}
  \right).
\]
Let $\mathcal{D}_{B}$ denote the labeled task data available to the
adversary, where $|\mathcal{D}_{B}|\le B$, and let
$\mathcal{Y}_{\mathrm{pub}}$ denote outputs available through the normal
inference interface. The adversary's complete view is
\[
  \mathsf{View}_{\mathcal{S}}^{T,B}
  :=
  \left(
    \mathcal{M}_{\textsf{pre}},
    \operatorname{pub}(\mathcal{S}),
    \mathcal{D}_{B},
    \{\mathsf{V}_{\ell,t}\}_{\ell,\,1\le t\le T},
    \mathcal{Y}_{\mathrm{pub}}
  \right),
\]
where $\operatorname{pub}(\mathcal{S})$ contains the public description
and parameters of the deployed scheme.
The victim weights, private fine-tuning data, obfuscation keys, recovery
state, and post-recovery values inside the TEE remain hidden.

\paragraph{Attack and output.}
Let $\mathcal{A}$ be any polynomial-time attack algorithm. Within the
budgets $T$ and $B$, $\mathcal{A}$ may jointly process all public
information and recorded views and may choose later protocol-valid
queries based on earlier observations. It outputs a surrogate model
$\widehat{\mathcal{M}}$ (the $\mathcal{M}_{\textsf{sur}}$ of
\S~\ref{sec:threat_model}) intended to reproduce the task behavior of
$\mathcal{M}_{\textsf{vic}}$:
\[
  \widehat{\mathcal{M}}
  \leftarrow
  \mathcal{A}\!\left(\mathsf{View}_{\mathcal{S}}^{T,B}\right),
  \qquad \mathcal{A}\in\mathsf{PPT}.
\]
The compact view notation includes the entire protocol-valid interaction;
it does not restrict the attack to a particular extraction method. The
adversary's goal is to maximize the surrogate's task accuracy.

\paragraph{Winning condition.}
After the attack terminates, the surrogate and victim models are
evaluated on the same evaluation set $\mathcal{D}_{\mathrm{eval}}$.
The attack succeeds if
\[
  \operatorname{Acc}_{\mathcal{D}_{\mathrm{eval}}}
  (\widehat{\mathcal{M}})
  \ge
  \operatorname{Acc}_{\mathcal{D}_{\mathrm{eval}}}
  (\mathcal{M}_{\mathrm{vic}})-\delta.
\]
Here, $\operatorname{Acc}_{\mathcal{D}_{\mathrm{eval}}}(\mathcal{M})$
denotes the task accuracy of $\mathcal{M}$ on
$\mathcal{D}_{\mathrm{eval}}$, and $\delta$ is the allowed accuracy gap
from the victim model. This criterion defines attack success; it does
not by itself assert that a particular obfuscation scheme prevents every
such attack.

\section{Proof of the Composition Theorem}
\label{app:composition-proof}

\begin{proof}[Proof of Theorem~\ref{thm:composition}]
Let $\mathcal{C}=\mathcal{Q}\circ\mathcal{P}$.
For an admissible input $X$ and weight $W$, let
$W_1=\mathcal{O}_1(W)$ and $W_2=\mathcal{O}_2(W_1)$. Correctness of
the two primitives gives
\[
  \mathcal{R}_2(XW_2)=XW_1,
  \qquad
  \mathcal{R}_1(XW_1)=XW.
\]
Applying the recoveries in reverse proves (i). Since the composition
executes each primitive once, the additive cost model gives
\[
  C(\mathcal{C})=C(\mathcal{P})+C(\mathcal{Q}),
\]
which proves (ii).

For (iii), let $\mathcal{S}\in\{\mathcal{P},\mathcal{Q}\}$ and assume
a polynomial-time full-transcript simulator under the same budgets such
that
\[
  \mathsf{Sim}_{\mathcal{S}}
  \bigl(\mathsf{View}_{\mathcal{S}}^{T,B}\bigr)
  \stackrel{d}{=}
  \mathsf{View}_{\mathcal{C}}^{T,B}.
\]
For any polynomial-time attack $\mathcal{A}$ on $\mathcal{C}$,
$\mathcal{A}\circ\mathsf{Sim}_{\mathcal{S}}$ is an attack on
$\mathcal{S}$ with the same surrogate-output distribution and hence the
same winning condition.
Thus, every successful attack on $\mathcal{C}$ induces one on
$\mathcal{S}$. Applying this argument to both $\mathcal{P}$ and
$\mathcal{Q}$ proves (iii).
\end{proof}

\section{Attack Strategy under Fully Refreshed Components}
\label{app:fully-refreshed-attack}

This appendix discusses an attack strategy for the alternative schedule
in which all four obfuscation components are refreshed across views:
\[
  \widetilde W_t=(W+S_t+L_t)D_t\Pi_t.
\]
The attack retains the same three-stage structure as \sysattack in the
main text: Stage~1 aligns the relative permutations, Stage~2 removes the
low-rank terms, and Stage~3 recovers the absolute column order and scales
before fine-tuning. The difference is that $S_t$ is no longer shared across
views and therefore also appears in cross-view differences. To handle this,
the attacker repeats each LDD query on different row subsets. Clean subsets
provide the normal rank used for permutation alignment, whereas subsets that
contain a nonzero entry of $S_t$ produce an additional rank increase. By
recording which sampled positions cause these increases, the attacker
estimates the support of each $S_t$ and excludes those entries when recovering
the low-rank terms.

\paragraph{Stage 1: permutation alignment and support localization.}
The original LDD test computes the rank of two column blocks using all
rows. Here, the attacker instead repeats each test on small random subsets of
rows. If a sampled block contains no nonzero entry from either $S_a$ or $S_b$,
the sample is clean and obeys the same rank--overlap relation used in
Stage~1: the rank reveals how many victim columns the two blocks share. If the
sample intersects a sparse entry, the additional perturbation generally
increases the rank. Repeating the test and taking the smallest stable rank
therefore recovers the clean LDD result. The seed, disambiguation, and
propagation procedures from the main attack can then be applied without
change to recover the relative permutation between the two views.

The same tests also reveal the sparse support. Whenever the observed
rank is larger than the clean rank, the sampled row--column positions receive
an anomaly vote. The attacker repeats this test with different row subsets
and different column blocks. Positions containing a sparse perturbation
receive such votes consistently, whereas clean positions do not. Aggregating
the votes across view pairs yields an estimated support $\widehat\Omega_t$ for
each $S_t$. Stage~1 thus outputs the relative permutations together with a
support estimate $\widehat\Omega_t\approx\suppop(S_t)$ for every view.

\paragraph{Stage 2: masked removal of the low-rank terms.}
After aligning the permutations, the attacker estimates the relative
column scales using only entries outside the estimated sparse supports. Let
$\widetilde W_t^\circ$ denote the resulting view in the anchor's column order
and scale. For two views $a$ and $b$, their difference is
\[
  \widetilde W_a^\circ-\widetilde W_b^\circ
  =(S_a-S_b+L_a-L_b)D_0.
\]
Let $\widehat\Omega_{ab}=\widehat\Omega_a\cup\widehat\Omega_b$. Outside
this set, the sparse difference disappears, leaving
\[
  \mathcal P_{\widehat\Omega_{ab}^{c}}
  (\widetilde W_a^\circ-\widetilde W_b^\circ)
  \approx
  \mathcal P_{\widehat\Omega_{ab}^{c}}((L_a-L_b)D_0).
\]
After excluding $\widehat\Omega_{ab}$, the difference contains only
the low-rank terms and reduces to the setting of the original Stage~2. The
original Stage~2 procedure then produces an estimate of $WD_0$ outside the
sparse support; the excluded entries are passed to Stage~3 for fine-tuning.

\paragraph{Stage 3: absolute alignment and fine-tuning.}
Stage~3 receives the Stage~2 estimate of $WD_0$ outside the sparse
support and directly reuses the original Stage~3 procedure for absolute
alignment, scale recovery, and fine-tuning. The attack therefore succeeds
under the fully refreshed schedule.
\end{document}